%% file: CorePeripheryModel.tex
\documentclass[11pt,a4paper]{article}

\usepackage{amsmath,amssymb}

\usepackage{amsthm}
\usepackage{hyperref}
\usepackage{mathtools}
\usepackage[shortlabels]{enumitem}
\usepackage{bbm}
\usepackage{graphicx,graphics,float}
\usepackage{subfigure}
\usepackage[left=1in, right=1in, top=1in, bottom=1in, includefoot, headheight=15pt]{geometry}
\usepackage{stmaryrd}

\usepackage{color}
\definecolor{darkgreen}{rgb}{0, .5, 0}
\definecolor{darkred}{rgb}{.5, 0, 0}

\usepackage{setspace}

\theoremstyle{plain}
\newtheorem{theorem}{Theorem}[section]
\newtheorem*{mocktheorem*}{Mock Theorem}
\newtheorem{proposition}[theorem]{Proposition}
\newtheorem{corollary}[theorem]{Corollary} 
 
\newtheorem{lemma}[theorem]{Lemma} 
\newtheorem{example}[theorem]{Example}
\theoremstyle{definition} 
\newtheorem{definition}[theorem]{Definition}

\numberwithin{equation}{section}
\usepackage{multicol}

\input{definitions.tex}

\usepackage{relsize}
\usepackage{bm}

\begin{document}
\singlespacing
\title{\LARGE\bf Financial Contagion in a Generalized Stochastic Block Model}
\author{Nils Detering\thanks{Department of Statistics and Applied Probability, University of California, Santa Barbara, CA 93106, USA. Email: detering@pstat.ucsb.edu}, Thilo Meyer-Brandis\thanks{Department of Mathematics, University of Munich, Theresienstra\ss{}e 39, 80333 Munich, Germany. Emails: meyerbra@math.lmu.de, kpanagio@math.lmu.de and ritter@math.lmu.de}, Konstantinos Panagiotou\footnotemark[2], Daniel Ritter\footnotemark[2]  }

\maketitle

\begin{abstract}
\noindent One of the most defining features of the global  financial network is its inherent complex and intertwined structure. From the perspective of systemic risk it is important to understand the influence of this network structure on default contagion. Using sparse random graphs to model the financial network, asymptotic methods turned out powerful to analytically describe the contagion process and to make statements about resilience. So far, however, they have been limited to so-called {\em rank one} models in which informally the only network parameter is the degree sequence (see \cite{Cont2016,Detering2016} for example) and the contagion process can be described by a one dimensional fix-point equation. These networks fail to account for a pronounced block structure such as core/periphery or a network composed of different connected blocks for different countries. We present a much more general model here, where we distinguish vertices (institutions) of different types and let edge probabilities and exposures depend  on the types of both, the receiving and the sending vertex plus additional parameters. Our main result allows to compute explicitly the systemic damage caused by some initial local shock event, and we derive a complete characterisation of resilient respectively non-resilient financial systems. This is the first instance that default contagion is rigorously studied in a model outside the class of rank one models and several technical challenges arise. Moreover, in contrast to previous work, in which networks could be classified as resilient or non resilient, independent of the distribution of the shock, information about the shock becomes important in our model and a more refined resilience condition arises. Among other applications of our theory we derive resilience conditions for the global network based on subnetwork conditions only.

\end{abstract}

\medskip
\noindent\textit{Keywords:} systemic risk, financial contagion, inhomogeneous random graphs, weighted random graphs, directed random graphs, stochastic block model, core-periphery, assortative random graphs, counterparty dependent exposures

\section{Introduction}
Modern financial networks of counterparty exposures are characterized by a huge intrinsic complexity. One example for this are cross-border relations that have emerged over the last decades in the course of globalization \cite{Chinazzi2013,Degryse2010,Halaj2013,Minoiu2013}. While the possibility of making business on an international level certainly provides benefits to single institutions by facilitating diversification and giving access to different 
markets \cite{Aghion2010,Aoki2010,Artis2007,Artis2012,Dell'ariccia2008,Faria2007}, the convoluted trans-border dependencies can also pose a great threat to the global system and increase \emph{systemic risk}. 
The financial crisis in the years 2007/08 is a striking example of how an initially locally confined shock -- the burst of the US housing bubble -- had a catastrophic impact on the economy worldwide. In addition to cross-border exposures connecting various subsystems, modern financial networks also exhibit a distinctive tiered structure -- usually referred to as \emph{core-periphery structure} -- that can arise in many different shapes; see \cite{Boss2004} for Austria, \cite{Cont2013} for Brazil, \cite{Craig2014} for Germany, \cite{Fricke2015} for Italy, \cite{Veld2014} for the Netherlands, \cite{Langfield2014} for the UK, and \cite{Alves2013} for the European interbank network. It is an important research question to model and understand how these complex network configurations influence the stability of financial systems.

\begin{figure}[t]
\centering
    \includegraphics[width=0.52\textwidth]{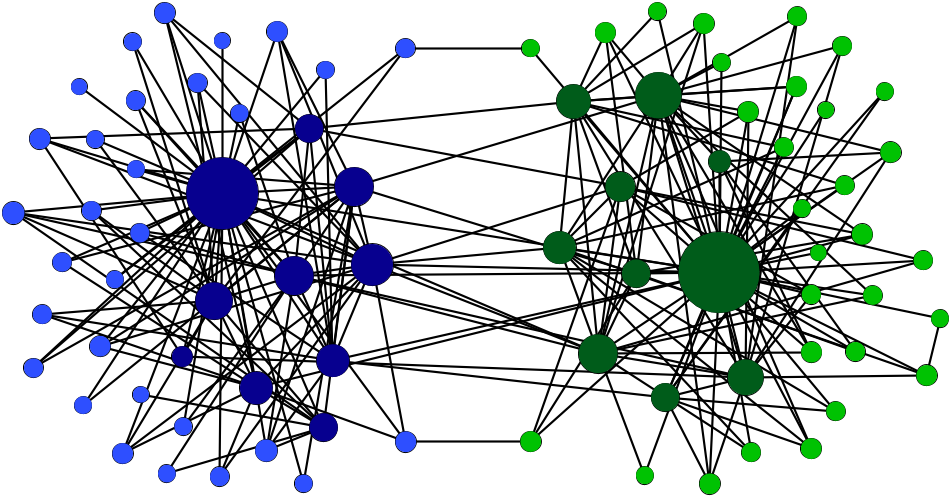}
\caption{A sample network consisting of two cores (darkblue resp.~darkgreen) and the associated peripheries (lightblue resp.~lightgreen). Vertex sizes correspond to the respective degrees. For simplicity the network is depicted undirected and the strength of  links is omitted. }
\label{fig:core:periphery:network}
\end{figure}

One important direct channel of systemic risk that we focus on here is \emph{default contagion}. Instead of considering a concrete system configuration, we use a random graph setting to generate a whole ensemble of possible financial networks and it is assumed that the observed network is a typical member of this family. This can be achieved by calibrating the random graph's distribution to macroscopic statistical properties of the observed network, such as 
the distribution of the degrees and the balance sheet exposures. Our results in this setting are then derived in terms of these global statistics only and hold for essentially the whole ensemble of random configurations. 
In particular, members of this family can vary significantly on a microscopic level, rendering this approach robust to local changes of the network or statistical uncertainties. Furthermore, under the premise that over time the network may change locally but keeps its global statistics, a fact that is supported empirically by \cite{Cont2013} in the context studied here, the results are applicable also for future systems; this is a desirable property from a regulator's point of view. Additionally, as the results for contagion in random graphs are typically derived in an asymptotic way for a large number of institutions, there is a naturally emerging notion of resilience of a network that does not depend on any arbitrarily chosen parameters such as confidence values. 

Within the literature on financial networks the random graphs perspective was first introduced in~\cite{Gai2010} and carried out for the configuration model in \cite{Cont2016}. In \cite{Hurd2016, Hurd2017} it has been extended to a version of the configuration model which allows to capture assortativity of the network. The works \cite{Detering2015a,Detering2016} analyze default contagion on inhomogeneous random graphs.
Other papers that investigate how the network structure influences contagion use for example matrix majorisation in \cite{Capponi2016} or Baysian methods in \cite{Chong2016}.

\vspace{-7pt}
\paragraph{Contribution} Our paper provides a systematic analysis of the effect of complex network structures on contagion and system stability by means of random graphs. We quantify the effect that different subsystems such as core/periphery, geographical location, institutional form (bank, insurance, etc.), have on contagion in the global system and we make combined effects of global network structures and local heterogeneity explicit. Our model setup accounts for the strong complexity observed in global, tiered financial systems in contrast to previous work which has depicted only the effect of heterogeneity, measured in terms of the degree distribution of the individual institutions. 

\vspace{2pt}
\emph{Block Model}~~To describe different subsystems (blocks) 
to every vertex we assign a type and vertices with the same type form a block. Institutions of the same type could be similar in terms of business strategy or they could reside in the same country for example. In addition we assign a vector of weights to each vertex (institution) in the graph that describes the tendency to develop random edges to vertices of a certain type and thus controls the local characteristics (degree) of the vertex. Figure~\ref{fig:core:periphery:network} illustrates a sample from our model with only two types (blue and green) and different weights (and as a result degrees) for vertices of the same type. Our model is a tradeoff to reflect the above mentioned characteristics of financial networks but at the same time preserve analytical tractability. In its undirected version our network skeleton model can be embedded in a general class of random graphs proposed in \cite{Bollobas2007} based on possibly uncountable number of types but without weights. While we believe that in the most general model setting described in \cite{Bollobas2007} a finite dimensional analytic description for the contagion processes is not possible, the combination of types and weights allows us to derive a full analytic description of the process with a multi but finite dimensional system of differential equations. To our knowledge this is the first paper that investigates such kind of process in a model outside the popular class of rank 1 models (or variants thereof) which lead to one dimensional fixpoint equations for the contagion process (see \cite{Cont2016,Detering2015a,Detering2016}). The resulting multivariate setting poses significant new challenges on the analytical side especially for deriving stopping criteria that ensure that contagion activity in one subsystem is not reinitiating activity in a part that has slowed down already. We mention that -- especially in the statistical physics literature -- one often refers to \emph{fitness} of a vertex instead of type, see \cite{Caldarelli2002,Gandy2017,GANDY2019193,Servedio2004} for example. 
As our model description strongly depends on the choice of vertex types (blocks), for calibration purposes one needs to draw on methods for detection of communities and core-periphery structure, which are well developed. See for instance \cite{Blondel2008,Clauset2004,Copic2009,Fortunato2016,Zhang2014,Zhao2011} respectively \cite{Holme2005,Rombach2017}. 

\vspace{2pt}
\emph{Exposure Distributions} The standard setting in the random-graph approach to default contagion is that exposure distributions only depend on the creditor banks. Just thinking of the example of a core-periphery network shows the limitations of this assumption, as liabilities between core-banks should be significantly higher than liabilities of a core-bank to a periphery-bank. In the present work we waive this assumption thus allowing for realistic exposure distributions. While this might seem like a minor additional feature, it in fact increases mathematical complexity significantly as variations in the banks exposures cannot be encoded in vertex features anymore as in previous instances in which the contagion process could be reduced to a threshold model (\cite{Cont2016,Detering2016}). The exposures now become an integral part of the contagion process and it is one of the technical contributions of this article to provide in a first instance an analysis of such a setting.

\vspace{2pt}
\emph{Analytic Results on the Final Default Fraction and Resilience}~~Our main result in this paper is a full analytic desciption of the default contagion process in our model. Our first main result Theorem \ref{thm:general:weights} identifies  bounds for the final fraction of defaulted institutions in a system hit by some arbitrarily specified initial shock.
The following informal statement summarizes the main insights from Theorem \ref{thm:general:weights}.
\begin{mocktheorem*}
Consider a financial network of size $n\in\N$ that was hit by some exogenous shock so that some of the banks are initially defaulted. Then under certain regularity assumptions there exist constants $0\leq k\leq K\leq 1$ such that

\[ k + o_p(1) \leq \frac{\emph{number of finally defaulted institutions}}{n} \leq K + o_p(1), \]
where $o_p(1)$ denotes a term that 
vanishes as $n$ becomes large. The constants $k,K$ can be computed explicitly in terms of the global network statistics. In most relevant settings lower and upper bound equal ($k=K$).
\end{mocktheorem*}
For large network sizes $n$, the theorem hence determines explicitly the fraction of finally defaulted institutions in the system. 
As it only depends on the global network statistics the result is robust to statistical uncertainties on an institution level and stays true also for future points in time as long as macro-statistical properties of the system do not change~much.


From a regulator's viewpoint, more than merely describing the contagion process for certain shock events, it is important to understand whether today's financial system is able to absorb possible future shocks. A common approach to this question in risk management is to choose certain parameters for the initial stress and quantiles of the final damage that are deemed admissible. Our approach for large networks, however, gives rise to a \emph{parameter-free} alternative: we call a system \emph{non-resilient} if the final fraction of defaulted institutions is lower bounded by a positive constant no matter how small the initial shock was. 

Building on our description of the final default fraction, a second main contribution of this paper is thus to derive explicit criteria to decide whether a certain system is (non-)resilient. In contrast to the previous literature (see ~\cite{Cont2016,Detering2015a,Detering2016}) where a system was classified either as resilient or non-resilient independent of the shock that is applied it turns out that in the model studied here the distribution of the initial shock becomes relevant. Different shocks may target different subsystems and the effect varies greatly; this allows us to quantify the effect of clusters and tiers for system stability more accurately than previous work and we can investigate which network characteristics promote the spread of distress and how global cascades can be contained. 

\vspace{2pt}
\emph{Applications}~~Recall that the leading question we want to study in this paper is what impact the complex intertwining of modern financial networks -- such as the global connectedness and the tiered structures -- has on their stability. We demonstrate by several applications how the features of our model effect the course of the contagion process. First, we consider the influence of a non-resilient subsystem on a lager global system. We find that not only this non-resilient subpart may make the system as a whole non-resilient, but it can also cause severe damage to a priori resilient parts of the network. More crucially, we also find that a system consisting only of resilient subsystems may become non-resilient due to cross-type (e.\,g.~cross-border) linkages. Therefore, we further derive criteria for resilience of the global system based on conditions for the subsystems and the connectivity between them. Finally, we consider the reshuffling of exposures in a network where originally all exposures only depend on the creditor -- the exposed party -- such that the new network's exposures respect a core-periphery structure. We show that by this procedure an initially resilient system can become non-resilient, which stresses the importance of creditor and debtor dependent exposures in a model for default contagion as otherwise contagion effects are underestimated.


\vspace{-7pt}
\paragraph{Outline} The paper is structured as follows: In Section~\ref{sec:default:fin}, we introduce the random graph and describe the model for default contagion. We then derive asymptotic results on the final default fraction in Section~\ref{sec:asymptotic:results}. In Section~\ref{sec:resilience} we discuss (non-)resilience and the corresponding criteria. We then present our applications in Section~\ref{sec:applications} and perform simulations for finite networks. Finally, we provide proofs of our results in Section~\ref{sec:proofs}.

\section{Default Contagion on a Random Multi-type Network}\label{sec:default:fin}
We describe a model for a financial network consisting of $n\in\N$ vertices (institutions) and random directed edges between them. We usually think of an institution $i\in[n]:=\{1,\ldots,n\}\subset\N$ as a bank in an interbank network and of a directed edge going from institution $i\in[n]$ to $j\in[n]$ as a financial exposure of $j$ which emanates from $i$, for example by an outstanding interbank-loan from $i$ to $j$. Our model accounts for two more features. First, we assign exposures to the edges representing the amount of the loan. As will be clear from the construction in the following, the assignment of exposures depends on both the creditor- and the debtor-institution. Second, we assign different types to the institutions in the network. This allows to describe more involved network structures such as core-periphery networks -- a two-type network in our terminology -- and (dis-)assortative structures.

\subsection{Vertex Types}
We begin by assigning to each institution $i\in[n]$ a type $\alpha_i\in[T]$, where $T\in\N$ is the fixed number of types. In the prominent case of a core-periphery network we choose $T=2$ and a bank $i\in[n]$ shall be a core bank if $\alpha_i=1$ resp.~a periphery bank if $\alpha_i=2$. Alternatively banks in one country or region could be assigned the same type. Hence the financial network is partitioned into sets of institutions of different types, which we also call blocks.

\subsection{Vertex Weights and Random Exposures}
Next, we fix $R\in\N$ and we construct a random network with exposures in $[R]$. To this end, assign to each institution $i\in[n]$ a set $\{w_i^{-,r,\alpha},w_i^{+,r,\alpha}\}_{1\leq r\leq R, 1\leq \alpha\leq T}$ of non-negative vertex-weights. Here the weight $w_i^{-,r,\alpha}$ describes the tendency of bank $i$ to develop incoming edges with exposure $r$ from institutions of type $\alpha$. Similarly, $w_i^{+,r,\alpha}$ describes the tendency of $i$ to form outgoing edges with exposure $r$ to institutions of type $\alpha$. To formalize this, 
let $X_{i,j}^r$ be the indicator random variable which is $1$ if there is an edge with exposure $r$ going from $i$ to $j$ and $0$ otherwise and let $X_{i,j}^r\sim\mathrm{Be}(p_{i,j}^r)$ be a Bernoulli random variable with  expectation
\begin{equation}\label{eqn:edge:prob}
p_{i,j}^r := \begin{cases} \min\{R^{-1}, n^{-1}\}w_i^{+,r,\alpha_j}w_j^{-,r,\alpha_i},&i\neq j,\\0,&i=j.\end{cases}
\end{equation}
To avoid multiple edges with different exposures between the institutions, we assume $\{X_{i,j}^r\}_{1\leq r\leq R}$ to be mutually exclusive in the sense that $\sum_{1\le r\le R} X_{i,j}^r\leq 1$. Also, we assume that edges between different pairs of institutions are independent, i.\,e.~$X_{i_1,j_1}^{r_1}\perp X_{i_2,j_2}^{r_2}$ for all $r_1,r_2\in[R]$ if $i_1\neq i_2$ or $j_1\neq j_2$. In particular, $X_{i,j}^{r_1}\perp X_{j,i}^{r_2}$ for all $r_1,r_2\in[R]$, $i\neq j$. This can for example be achieved by introducing a sequence of independent random variables $U_{i,j}$, each distributed uniformly on the interval $[0,1]$, and letting $X_{i,j}^r=\1\big\{U_{i,j}\in\big[\sum_{s\leq r-1}p_{i,j}^s,\sum_{s\leq r}p_{i,j}^s\big)\big\}$. The upper bound $R^{-1}$ in \eqref{eqn:edge:prob} then ensures that $\sum_{s\leq R}p_{i,j}^s\leq1$. 


\subsection{Capital and Default Contagion}
We assign to each institution $i\in[n]$ an initial amount of capital (equity) $c_i\in\N_0\cup\{\infty\}$. We call an institution solvent if $c_i>0$ and insolvent if $c_i=0$ and we denote by $\mathcal{D}_0:=\{i\in[n]\,:\,c_i=0\}$ the set of initially defaulted institutions. The initial default shall be due to some exogenous event such as a stock market crash. Because of the interconnections in the network the default of the institutions in $\mathcal{D}_0$ will spread through the network. This happens since the defaulted banks cannot (fully) repay their loans to their creditors. As first suggested in \cite{Gai2010} it is a reasonable assumption that defaulted debtors cannot repay any of their debts since processing their default may take months or even years while financial contagion is a short term process. In fact one can generalize our model to the case of a fixed constant recovery rate simply by adjusting the capitals. The default contagion process can then be described as follows. In round $k\geq 1$ of the default cascade the set of defaulted institutions is 

\vspace*{-0.1cm}
\begin{equation}
\label{eqn:default:contagion}
\mathcal{D}_k := \Bigg\{ i\in[n]\,:\,c_i\leq\sum_{r \in [R]} r\sum_{j\in\mathcal{D}_{k-1}}X_{j,i}^r\Bigg\}.
\end{equation}

\vspace*{-0.1cm}
\noindent In particular, $\mathcal{D}_0\subseteq\mathcal{D}_1\subseteq\cdots$ and the chain of default sets stabilizes at round $n-1$ the latest. We hence denote the final default set by $\mathcal{D}_n:=\mathcal{D}_{n-1}$. Note that the only randomness in this process stems from the random links in the network. Once a network configuration has been fixed, the whole default contagion sequence is fully determined.
The following two remarks address interesting possible extensions of our setting. 


\begin{remark} \label{RemSysImp}
In the following we use the fraction of finally defaulted institutions $n^{-1}\vert\mathcal{D}_n\vert$ 
as a measure for the severity of a given shock in the network.
It makes sense, however, to not only consider the
finally defaulted institutions but to introduce a parameter $s_i\in\R_{+,0}$ of general systemic importance for each bank $i\in[n]$ and use $\mathcal{S}_n:=n^{-1} \sum_{i\in\mathcal{D}_n}s_i$ to measure the severity of the shock. Since the systemic importance values $s_i$ do not influence the process \eqref{eqn:default:contagion}, all our results can be extended to this more general setting under mild assumptions on $\{s_i\}_{i\in[n]}$. See Remark \ref{rem:systemic:importance} for the general setting.
\end{remark}

\begin{remark}
To model realistic financial networks with general exposure values, it would be a priori necessary to choose $R$ very large and our model would become high-dimensional. Instead of considering an exposure $r$ between two institutions, however, one can also interpret it as a more general factor of impact. It is then possible to model unbounded exposure distributions also with a considerably small choice of $R$. More precisely, we can assign to each institution $i\in[n]$ a list of exchangeable random variables and take the sum over $r$-many of them to compute the actual exposure of an edge with impact $r$. Under minor assumptions on the distribution of these exchangeable random variables, the main results of this paper can be generalized to this setting by a conditioning argument. See \cite{Detering2016} for a precise discussion of this idea in the case $R=T=1$.

\end{remark}


\subsection{Regular Vertex Sequences}
In the previous subsections we introduced several parameters and in particular, any random ensemble is described by the weight sequences $\bm{w}^{-,r,\alpha}:=(w_1^{-,r,\alpha},\ldots,w_n^{-,r,\alpha})$ and $\bm{w}^{+,r,\alpha}:=(w_1^{+,r,\alpha},\ldots,w_n^{+,r,\alpha})$ for $r\in[R]$ and $\alpha\in[T]$, 
the capital sequence $\bm{c}:=(c_1,\ldots,c_n)$, and the 
type sequence $\bm{\alpha}:=(\alpha_1,\ldots,\alpha_n)$. That is, much of the information about the system -- in particular the structure of the network configuration -- is contained in the empirical distribution function

\vspace*{-0.5cm}
\[ F_n(\bm{x},\bm{y}, \ell, m) := n^{-1}\sum_{i\in[n]}~\prod_{\substack{r\in [R],\alpha\in [T]}}\1\left\{w_i^{-,r,\alpha}\leq x^{r,\alpha},w_i^{+,r,\alpha}\leq y^{r,\alpha}\right\} \1\left\{c_i\leq \ell, \alpha_i\leq m\right\}, \]

\vspace*{-0.27cm}
\noindent for $\bm{x},\bm{y}\in\R^{[R]\times[T]}$, $\ell\in\N_0\cup\{\infty\}$ and $m\in[T]$. 

The setting described so far puts us in the position to model a system with a given number~$n$ of institutions. However, as already described in the introduction, our main focus is in studying how the structure in the underlying network affects the contagion process and, more generally, the (in-)stability of the system as a whole. Towards this aim we proceed as follows. Instead of restricting our attention  to a single system configuration, we consider an ensemble of systems that are similar, where the structural similarity is measured precisely in terms of the joint empirical distribution of all parameters that we consider. In particular, we assume that we have a collection of systems with a varying number $n$ of
institutions with the property that the sequence $(F_n)_{n \in \mathbb{N}}$ of empirical distributions converges.

\begin{definition}[Regular vertex sequence]\label{def:regular:vertex:sequence}
A sequence $(\bm{w}^{-,r,\alpha}(n),\bm{w}^{+,r,\alpha}(n), \bm{c}(n), 
\bm{\alpha}(n))_{n\in\N}$ of model parameters for different network sizes $n\in\N$ is called a \emph{regular vertex sequence} if the following conditions are satisfied. Note that we do not specify the ranges of $r$ and $\alpha$ to keep the notation concise. Here and in the following we always mean $r\in[R]$ and $\alpha\in[T]$.
\begin{enumerate}[(a)]
\item \textbf{Convergence in distribution:} 
Let $(W_n^{-,r,\alpha},W_n^{+,r,\alpha},C_n, A_n)$ be a random vector distributed according to the empirical distribution function $F_n$. Then there exists a distribution function $F$ such that $F_n(\bm{x},\bm{y},\ell,m)\to F(\bm{x},\bm{y},\ell,m)$ for all points $(\bm{x},\bm{y},\ell,m)$ at which $F_{\ell,m}(\bm{x},\bm{y}):=F(\bm{x},\bm{y},\ell,m)$ is continuous. Denote by $(W^{-,r,\alpha},W^{+,r,\alpha},C,A)$ a random vector distributed according to 
$F$.
\item \textbf{Convergence of average weights:} For $n\to\infty$  
and all $1\leq r\leq R$, $1\leq\alpha\leq T$:
\[ \E[W_n^{-,r,\alpha}] \to \E[W^{-,r,\alpha}]<\infty \quad\text{and}\quad \E[W_n^{+,r,\alpha}] \to \E[W^{+,r,\alpha}]<\infty. \]
\end{enumerate}
\end{definition}

To understand the asymptotic degree sequence of our model let us fix $r\in [R]$ and only look at links with exposure size $r$. We look at a vertex $i\in [n]$ of type $\beta$ with out-weight to vertices of type $\alpha$ equal to $w_i^{+,r,\alpha}$. Then its out-degree towards vertices of type $\alpha$ is a sum of independent Bernoulli trials, whose expectation converges to $w_i^{+,r,\alpha} \E [W^{-,r,\alpha}\1\{A=\beta\}]$ due to property (b) of Definition~\ref{def:regular:vertex:sequence}. With this observation similar ideas for a Poisson coupling of the degree sequence as used in \cite[Theorem 3.3]{Detering2015a} can be applied to prove the following result about the degree sequence in our financial system.
\begin{proposition}
Consider a financial system described by a regular vertex sequence and let $D_n^{\pm,r,\alpha}$ be the $r$-out/in-degree with respect to banks of type $\alpha$ of some bank in the network chosen uniformly at random. 
Then for each $(r,\alpha)\in[R]\times[T]$, as $n\to\infty$, in distribution
\begin{equation}\label{eqn:degree:convergence}
D_n^{\pm,r,\alpha} \to \mathrm{Poi}\Bigg(W^{\pm,r,\alpha}\sum_{\beta\in[T]}\zeta^{r,\beta,\alpha}_\mp\1\{A=\beta\}\Bigg),
\end{equation}
where $\zeta^{r,\beta,\alpha}_\mp:=\E\left[W^{\mp,r,\beta}\1\{A=\alpha\}\right]$.
In particular, the $r$-out/in-degree with respect to banks of type $\alpha$ of some uniformly chosen bank of type $\beta$ converges in distribution to $\mathrm{Poi}\big(W^{\pm,r,\alpha}\zeta_\mp^{r,\beta,\alpha}\big)$. 
\end{proposition}

\section{Asymptotic Results}\label{sec:asymptotic:results}
This section presents results for the final default fraction $n^{-1}\vert\mathcal{D}_n\vert$ triggered by some set of initial defaults, which is one of the main contributions of this paper, see Subsection~\ref{subsecFinDef}. In Subsection~\ref{subsecPrel}, we first introduce some functions that will play an important role. For the sake of readability, most of the -- mainly technical -- proofs of this section are moved to the appendix.

\subsection{Preliminaries}\label{subsecPrel}
Denote in the following $V:=[R]\times[T]^2$. Let $\psi_\ell(x_1,\ldots,x_R) := \P\left(\sum_{s\in[R]}sX_s\geq \ell\right)$ for independent
Poisson random variables $X_s \sim \mathrm{Poi} (x_s)$.  In the following, functions $f^{r,\alpha,\beta}:\R_{+,0}^V\to\R$, $(r,\alpha,\beta)\in V$,  and $g:\R_{+,0}^V\to\R_{+,0}$ will play a central role. They are given by
\begin{equation*}
\begin{gathered}
f^{r,\alpha,\beta}(\bm{z}) = \E\Bigg[W^{+,r,\alpha}\psi_C\Bigg(\sum_{\gamma\in[T]}W^{-,1,\gamma}z^{1,\beta,\gamma},\ldots,\sum_{\gamma\in[T]}W^{-,R,\gamma}z^{R,\beta,\gamma}\Bigg)\1\{A=\beta\}\Bigg] - z^{r,\alpha,\beta},\\
g(\bm{z}) = \sum_{\beta\in[T]}\E\Bigg[\psi_C\Bigg(\sum_{\gamma\in[T]}W^{-,1,\gamma}z^{1,\beta,\gamma},\ldots,\sum_{\gamma\in[T]}W^{-,R,\gamma}z^{R,\beta,\gamma}\Bigg)\1\{A=\beta\}\Bigg].
\end{gathered}
\end{equation*}

Let us provide some intuition about the meaning of these functions. The contagion process described in (\ref{eqn:default:contagion}) can be restated in an equivalent sequential form where in each step only the effect of one defaulted institution on the system is explored. Defaulted institutions can then have two status: {\em explored} meaning their effect on the system has already been accounted for and {\em unexplored} meaning that their effect on the system has not yet been accounted for. We look at the case $T=R=1$ first. Then there is only one function $f=f^{1,1,1}$ of one variable $z=z^{1,1,1}$ which shall denote the total out-weight of defaulted and explored institutions divided by $n$. For an institution $i\in [n]$ with in-weight $w_i^{-,1,1}$ we know that in the limit ($n\rightarrow \infty$) the number of neighbors is given by $ \mathrm{Poi} (w_i^{-,1,1} \zeta ) $ with $\zeta= \zeta ^{1,1,1}_+$ the total out-weight divided by $n$, i.e. $\zeta:=\E\left[W^{+,1,1}\right] = n^{-1} \sum_{i \in [n]} w_i^{+,1,1} + o(1)$. Replacing $\zeta$ by $z$ the number of defaulted and explored neighbours  should intuitively be given by $\mathrm{Poi} (w_i^{1,1} z )$. Continuing with this heuristics, the probability that $i$ has defaulted (no matter if already exposed or unexposed) is $\mathbb{P} (\mathrm{Poi} (w_i^{-,1,1} z ) \geq c_i)= \psi_{c_i} (w_i^{1,1} z)$ and summing up over $i\in [n]$, the out-weight of all defaulted institutions as a result of exposing defaulted institutions with out-weight roughly equal $nz$ is given by $\E [ W^{+,1,1} \psi_{c_i} (W^{-,1,1} z) ]$. If $\E [ W^{+,1,1} \psi_{c_i} (W^{-,1,1} z) ]=z$ or equivalently $f(z)=0$ then there are no unexposed defaulted institutions and the process might come to a standstill.

In the multivariate case the variable $z^{r,\beta,\gamma}$ equals the sum of the out-weights of defaulted and already explored type $\gamma$ vertices for building edges with exposure $r$ to type $\beta$ vertices. For given $\bm{z}\in \R_{+,0}^V$ the function value $f^{r,\alpha,\beta}(\bm{z})$ quantifies the total out-weight of defaulted but unexplored type $\beta$ vertices for building edges with exposure $r$ towards vertices of type $\alpha$. Similarly as above when $f^{r,\alpha,\beta}(\bm{z})=0$ the total out-weight of defaulted and unexplored type $\beta$ vertices for building edges with exposure $r$ to type $\alpha$ vertices is zero  meaning that infection from $\beta$ vertices to $\alpha$ vertices via exposure of size $r$ has slowed down but there can still be contagion activity in the graph. Only when $f^{r,\alpha,\beta}(\bm{z})=0$ for all $(r,\alpha,\beta)\in V$ the process might come to a standstill. Let $\bm{z}$ be such a joint root. Then, in the limit ($n\rightarrow \infty$), a vertex $i\in [n]$ of type $\beta$ with capital $c_i$ has r-in-degree from type $\gamma$ vertices approximately Poisson distributed with mean $w_i^{-,r,\gamma} \hat{z}^{r,\beta,\gamma}$. Vertex $i$ thus defaults with probability $\psi_{c_i} (\sum_{\gamma \in [T]} w_i^{-,1,\gamma} \hat{z}^{1,\beta,\gamma} ,\ldots,\sum_{\gamma \in [T]} w_i^{-,R,\gamma} \hat{z}^{R,\beta,\gamma} )$. If we now choose the vertex $i\in [n]$ uniformly instead, then its default probability is given by $g(\hat{\bm{z}})$. 

Before making this intuition rigorous and proving our main results we begin by investigating some basic but important properties of these functions.
\begin{lemma}\label{lem:properties:f}
The functions $f^{r,\alpha,\beta}(\bm{z})$, $(r,\alpha,\beta)\in V$, and $g(\bm{z})$ are continuous at all $\bm{z}\in\R_{+,0}^V$. Further, each function $f^{r,\alpha,\beta}(\bm{z})$ is monotonically increasing in all of its coordinates except $z^{r,\alpha,\beta}$.
\end{lemma}
\begin{proof}
Continuity 
follows from Lebesgue's dominated convergence theorem noting that the integrands are continuous in $\bm{z}$ and bounded by the integrable random variable $W^{+,r,\alpha}$ resp.~by $1$. Monotonicity of $f^{r,\alpha,\beta}$ follows directly from the monotonicity of the Poisson-probabilities.
\end{proof}

%

Let now $S:=\bigcap_{(r,\alpha,\beta)\in V}\{\bm{z}\in\R_{+,0}^V\,:\,f^{r,\alpha,\beta}(\bm{z})\geq0\}$. Since $f^{r,\alpha,\beta}(\bm{z})<0$ for any $\bm{z}\in\R_{+,0}^V$ with $z^{r,\alpha,\beta}>\E[W^{+,r,\alpha}\1\{A=\beta\}]$ and $S$ is an intersection of closed sets, $S$ is in fact compact. Note that clearly $\bm{0}\in S$. In general $S$ might consist of several disjoint, compact, connected components. Let in the following $S_0$ denote the component (i.\,e.~the largest connected subset) of $S$ containing $\bm{0}$. Since $S$ is a compact subset of $\R_{+,0}^V$, so is $S_0$. Define now $\bm{z}^*\in\R_{+,0}^V$ by $(z^*)^{r,\alpha,\beta}:=\sup_{\bm{z}\in S_0}z^{r,\alpha,\beta}$. The following lemma shows that in fact $\bm{z}^*\in S_0$ and it can hence be thought of as the maximal point of $S_0$. Further it identifies $\bm{z}^*$ as a joint root of all the functions $f^{r,\alpha,\beta}$, $(r,\alpha,\beta)\in V$, and shows the existence of a smallest joint root $\hat{\bm{z}}$. It will turn out later (see in particular Theorem \ref{thm:general:weights}) that the final fraction $n^{-1}\vert\mathcal{D}_n\vert$ is intimately related to these two (typically coinciding) joint roots of the functions $f^{r,\alpha,\beta}$, $(r,\alpha,\beta)\in V$. 

\begin{lemma}\label{lem:existence:hatz}
There exists a smallest joint root $\hat{\bm{z}}\in\R_{+,0}^V$ of all the functions $f^{r,\alpha,\beta}$, $(r,\alpha,\beta)\in V$, in the sense that $\hat{\bm{z}}\leq\bar{\bm{z}}$ componentwise for all joint roots $\bar{\bm{z}}$. Further, $\bm{z}^*$ as defined above is a joint root of the functions $f^{r,\alpha,\beta}$, $(r,\alpha,\beta)\in V$, and both $\hat{\bm{z}}\in S_0$ and $\bm{z}^*\in S_0$.
\end{lemma}

The lemma identifies $\bm{z}^*$ as 
the maximal joint root of $f^{r,\alpha,\beta}$, $(r,\alpha,\beta)\in V$, in $S_0$. 
However, if $S_0\subsetneq S$, then there will exist joint roots $\tilde{\bm{z}}\not\in S_0$ such that $\bm{z}^*\leq\tilde{\bm{z}}$ componentwise.

Often $\hat{\bm{z}}$ and $\bm{z}^*$ will coincide and then Theorem \ref{thm:general:weights} below will show that the final default fraction $n^{-1}\vert\mathcal{D}_n\vert$ converges to $g(\hat{\bm{z}})$ in probability. But in some pathological situations this is not the case and Theorem \ref{thm:general:weights} will yield a lower bound on $n^{-1}\vert\mathcal{D}_n\vert$ in terms of $\hat{\bm{z}}$ and an upper bound in terms of $\bm{z}^*$. Figures \ref{fig:one:joint:root} and \ref{fig:two:joint:roots} show two-dimensional examples of $f$. In both examples, we chose $R=2$ and $T=1$. In the first example, we further chose all weights to be $1$ and the capital of each bank to be $3$ with probability $80\%$ respectively $0$ with probability $20\%$. The functions $f^1(z^1,z^2):=f^{1,1,1}(z^1,z^2)$ and $f^2(z^1,z^2):=f^{2,1,1}(z^1,z^2)$, where $z^1:=z^{1,1,1}$ and $z^2:=z^{2,1,1}$, then have a unique joint root, i.\,e.~$\hat{\bm{z}}=\bm{z}^*$. In the second example, we chose all weights to be $2$ and the capital of each bank to be $3$ with probability $\approx 94.14\%$ respectively $0$ with probability $\approx 5.86\%$. In this case, there exist two distinct joint roots $\hat{\bm{z}}\neq\bm{z}^*$ in $S_0$. At $\hat{\bm{z}}$, the root sets of $f^1$ and $f^2$ do not cross each other but only touch.

\begin{figure}[t]
    \hfill\subfigure[]{\includegraphics[width=0.4\textwidth]{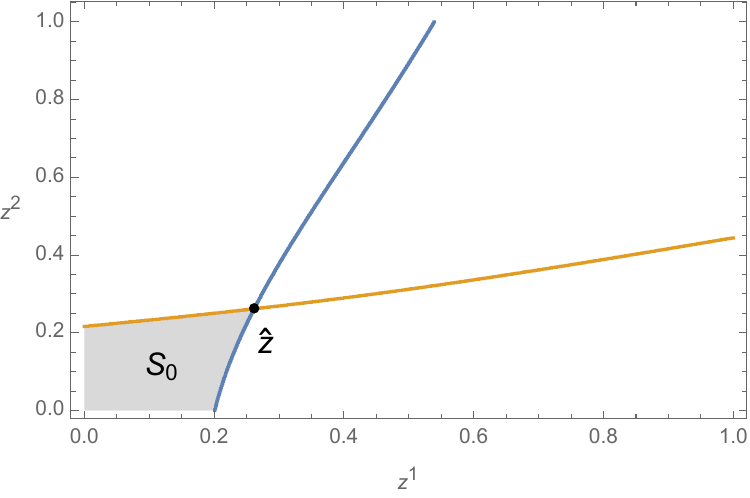}\label{fig:one:joint:root}}
    \hfill\subfigure[]{\includegraphics[width=0.4\textwidth]{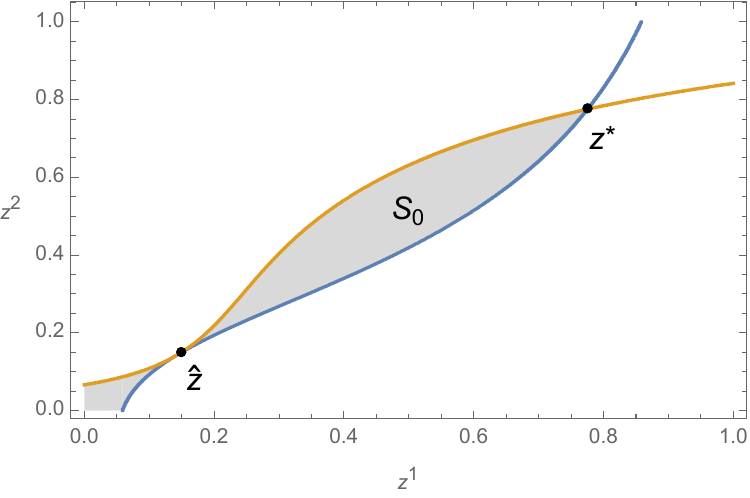}\label{fig:two:joint:roots}}\hfill
    
    \vspace*{-0.25cm}
\caption{Plot of the root sets of the functions $f^1(z^1,z^2)$ (blue) and $f^2(z^1,z^2)$ (orange) for two different example networks.}
\end{figure}

The next lemma provides two sufficient criteria to check if a joint root, such as $\hat{\bm{z}}$, equals $\bm{z}^*$. These depend on the (weak) directional derivatives of $f^{r,\alpha,\beta}$, $(r,\alpha,\beta)\in V$, and are hence natural extensions of the stable fixed point assumption in previous literature such as \cite{Cont2016,Detering2015a,Detering2016}. 

\begin{lemma}\label{lem:sufficient:criteria:z:star}
If $\bar{\bm{z}}\in S_0$ is a joint root of the functions $f^{r,\alpha,\beta}$, $(r,\alpha,\beta)\in V$, then $\bar{\bm{z}}=\bm{z}^*$ if one of the following holds:
\begin{enumerate}[(a)]
\item There exists $\bm{v}\in\R_+^V$ such that for all $(r,\alpha,\beta)\in V$ the directional derivatives $D_{\bm{v}}f^{r,\alpha,\beta}(\bar{\bm{z}})$ exist and $D_{\bm{v}}f^{r,\alpha,\beta}(\bar{\bm{z}})<0$.
\item There exist $\bm{v}\in\R_+^V$, $\kappa<1$ and $\Delta>0$ such that for every $\delta\in(0,\Delta)$,

\vspace*{-0.6cm}
\begin{align*}
\kappa v^{r,\alpha,\beta} &\geq \sum_{r'\in[R]}\E\Bigg[W^{+,r,\alpha}\Bigg(\sum_{\beta'\in[T]}v^{r',\beta,\beta'}W^{-,r',\beta'}\Bigg)\1\{A=\beta\}\\
&\hspace{1.5cm} \times\P\Bigg(\sum_{s\in[R]}s\mathrm{Poi}\Bigg(\sum_{\gamma\in[T]}W^{-,s,\gamma}\left(\bar{z}^{s,\beta,\gamma}+\delta v^{s,\beta,\gamma}\right)\Bigg)\in\{C-r',\ldots,C-1\}\Bigg)\Bigg].
\end{align*}
\end{enumerate}
\end{lemma}

\subsection{The Main Result for the Final Default Fraction}\label{subsecFinDef}
We now provide an asymptotic formula for the final default fraction $n^{-1}\vert\mathcal{D}_n\vert$ in terms of function $g$ and the joint roots $\hat{\bm{z}}$ and $\bm{z}^*$. The one type and exposure $1$ case ($R=T=1$) of the following result was obtained in \cite[Theorem 2.3]{Detering2015a} under an additional differentiability assumption guaranteeing that $\bm{z}^*=\hat{\bm{z}}$. 
\begin{theorem}\label{thm:general:weights}
Consider a financial system described by a regular vertex sequence and let $\hat{\bm{z}}$ and $\bm{z}^*$ be the smallest respectively largest joint root in $S_0$ of the functions $\{f^{r,\alpha,\beta}\}_{(r,\alpha,\beta)\in V}$. Then
\[ g(\hat{\bm{z}}) + o_p(1) \leq n^{-1}\vert\mathcal{D}_n\vert \leq g(\bm{z}^*) + o_p(1). \]
In particular, if $\hat{\bm{z}}=\bm{z}^*$, then $n^{-1}\vert\mathcal{D}_n\vert = g(\hat{\bm{z}}) + o_p(1)$.
\end{theorem}
See Section \ref{ssec:proof:main:general} for the proof.
\noindent 

\begin{remark}
Theorem \ref{thm:general:weights} determines $g(\hat{\bm{z}})$ as a lower bound on the fraction of all finally defaulted banks in the network. In fact, $g(\bm{z})$ is given by a sum over all the different types $\beta\in[T]$ in the network and it is thus no surprise that by small changes in the proofs of Theorems \ref{thm:finitary:weights} and \ref{thm:general:weights}, one derives that the number of finally defaulted banks of type $\beta$ is lower~bounded~by

\vspace*{-0.05cm}
\[ n\E\Bigg[\P\Bigg(\sum_{s\in[R]}s\mathrm{Poi}\Bigg(\sum_{\gamma\in[T]}W^{-,s,\gamma}\hat{z}^{s,\beta,\gamma}\Bigg)\geq C\Bigg)\1\{A=\beta\}\Bigg] + o_p(n). \]

\vspace*{-0.05cm}
\noindent The same reasoning allows to derive an upper bound in terms of $\bm{z}^*$.
\end{remark}

\begin{remark}\label{rem:systemic:importance}
As mentioned in Remark~\ref{RemSysImp}, in many situations it makes sense to measure the damage caused by a shock event to the whole system not simply by counting finally defaulted banks but to weight them according to some systemic importance value $s_i$, $i\in[n]$. If 
the empirical distribution of $\{s_i\}_{i\in[n]}$ converges in distribution and mean to a random variable $S$, then in analogy to Theorem \ref{thm:general:weights}, the total systemic importance of finally defaulted banks $\mathcal{S}_n:=\sum_{i\in\mathcal{D}_n}s_i$ is lower bounded by $n (g_S(\hat{\bm{z}})+o_p(1))$ and upper bounded by $n (g_S(\bm{z}^*)+o_p(1))$, where
\[ g_S(\bm{z}) := \sum_{\beta\in[T]}\E\Bigg[S\P\Bigg(\sum_{s\in[R]}s\mathrm{Poi}\Bigg(\sum_{\gamma\in[T]}W^{-,s,\gamma}z^{s,\beta,\gamma}\Bigg)\geq C\Bigg)\1\{A=\beta\}\Bigg]. \]
\end{remark}

\section{Resilient and Non-Resilient Networks}\label{sec:resilience}

In the previous section we derived results that allow us to determine the typical final default fraction  in large financial systems caused by an exogenous shock. Another important question from a regulator's point of view that we study in this section is whether a given system in an \emph{initially} unshocked state is likely to be resilient to small shocks or susceptible to default~cascades.

Note that for some fixed financial network $(W^{-,r,\alpha},W^{+,r,\alpha},C,A)$ all information about the initial shock stems from $C$ and by ``initially unshocked'' we mean that $c_i>0$ for all $i\in[n]$. We model small shocks to the system by an \emph{ex post infection} in the following sense: we introduce indicators $m_i\in\{0,1\}$, $i\in[n]$, 
with the meaning that (the initially solvent) bank $i$ becomes insolvent if $m_i=0$. This amounts to setting its capital to $c_im_i$. In analogy to Definition \ref{def:regular:vertex:sequence} we assume regularity of $\{m_i\}_{i\in[n]}$ (jointly with the rest of the parameters) and we denote by $M$ the limiting random variable of ex post infection. 
In particular, the financial system 
shall be described by the random vector $(W^{-,r,\alpha},W^{+,r,\alpha},C,A,M)$ with $\P(C=0)=0$ and $\P(M=0)>0$. Denote by $\mathcal{D}_n^M$ the random final default set that $M$ triggers and by $(f^M)^{r,\alpha,\beta}$, $g^M$ and $(\bm{z}^*)^M$ the analogues of $f^{r,\alpha,\beta}$, $g$ respectively $\bm{z}^*$ with $C$ replaced by $CM$. 

From a regulator's point of view a desirable property of a financial system is the ability to absorb small local shocks $M$ without larger parts of the system being harmed. In our asymptotic setting, we can even choose $M$ arbitrarily small and we call a system \emph{resilient} 
if the final default fraction $n^{-1}\vert\mathcal{D}_n^M\vert$ tends to $0$ as $\P(M=0)\to0$. If on the other hand $n^{-1}\vert\mathcal{D}_n^M\vert$ is lower bounded by some positive constant, we call the system \emph{non-resilient} (see Definition \ref{def:non:resilience} below). We say that a sequence of events $(E_n)_{n\in\N}$ holds with high probability (w.\,h.\,p.) if $\P(E_n)\to1$, as $n\to\infty$.
\begin{definition}[Resilience]\label{def:resilience}
A financial system is said to be \emph{resilient} if for each $\epsilon>0$ there exists $\delta>0$ such that for all $M$ with $\P(M=0)<\delta$ it holds $n^{-1}\vert\mathcal{D}_n^M\vert \leq \epsilon$ w.\,h.\,p.
\end{definition}

It will turn out that the resilience of the system strongly depends on the form of the set $S_0$ which was introduced in Subsection~\ref{subsecPrel}.
Our first result is a criterion guaranteeing resilience. 
\begin{theorem}[Resilience Criterion]\label{thm:resilience}
Consider a financial system described by a regular vertex sequence and assume that 
$S_0=\{\bm{0}\}$. Then 
the system is resilient.
\end{theorem}

In particular for $R=T=1$, resilience is ensured if $0=\inf\{z>0\,:\,f^{1,1,1}(z)<0\}$ and Theorem~\ref{thm:resilience} extends the one dimensional setting studied in \cite[Theorem 2.7]{Detering2015a}. Moreover, by Lemma~\ref{lem:sufficient:criteria:z:star}, if for some $\bm{v}\in\R_+^V$, $D_{\bm{v}}f^{r,\alpha,\beta}(\bm{0})$ exists and is negative for each $(r,\alpha,\beta)\in V$, then $S_0=\{\bm{0}\}$ and Theorem \ref{thm:resilience} is applicable.

Figure \ref{fig:resilient} shows a two-dimensional example satisfying the condition in Theorem \ref{thm:resilience}. We chose $R=2$, $T=1$, $W^{\pm,1}=W^{\pm,2}=1$ and $C=3$ and let $z^1:=z^{1,1,1}$ and $z^2:=z^{2,1,1}$. It can be seen from the figure that a small shock (here $5\%$ of all banks) keeps the smallest joint root of $f^1(z^1,z^2)=f^{1,1,1}(z^1,z^2)$ and $f^2(z^1,z^2)=f^{2,1,1}(z^1,z^2)$ for the shocked system close to $(0,0)$. By continuity properties of $g^M$ as $\P(M=0)\to0$ and $g(0,0)=0$ it follows that the final default fraction $n^{-1}\mathcal{D}_n^M$ is small for shocks $M$ with $\mathbb{P} (M=0)$ small.

\begin{figure}[t]
	\hfill\includegraphics[width=0.5\textwidth]{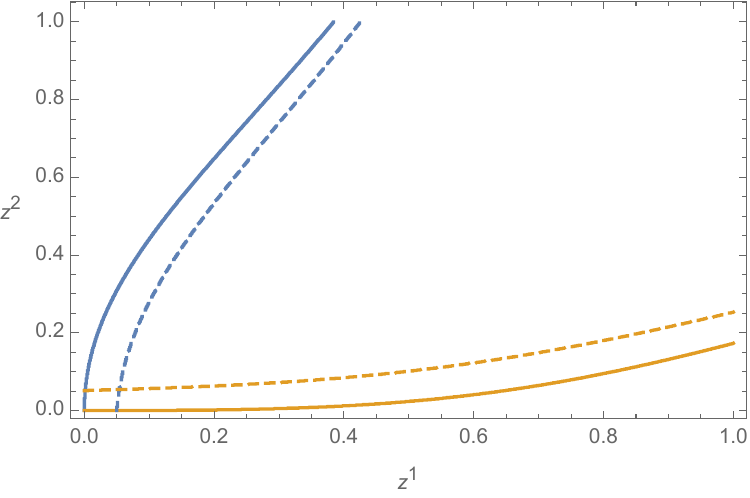}\hfill
		\caption{Plot of the root sets of the functions $f^1(z^1,z^2)$ (blue) and $f^2(z^1,z^2)$ (orange) for a financial system satisfying the condition in Theorem \ref{thm:resilience}. Solid: the unshocked functions. Dashed: the shocked functions.}\label{fig:resilient}
\end{figure}

On the other hand, to characterize non-resilient networks, a difficulty that arises is that the ex post shock $M$ possibly targets only certain subnetworks or only vertices with certain characteristics (weights). The subnetwork (for example core or periphery) that is targeted is encoded in the joint distribution of $(A,M)$ and it is not surprising that the effect of the ex-post infections depends on $\mathbb{P} (A=\alpha, M=0) = \mathbb{E} [\1\{A=\beta\}\1\{M=0\}]$. However, it turns out that this information is not sufficient in order to determine the impact of the shock $M$ completely and also the weights of targeted institutions matter, i.e. whether the targeted institutions build edges with exposure $r \in [R]$ to vertices of type $\beta \in [T]$. More precisely, define 
$$\tilde{V}:= \{ (r, \alpha , \beta ) \in V : \E[W^{+,r,\alpha}\1\{A=\beta\}]>0 \} $$
and note that $(r, \alpha , \beta ) \in \tilde{V}$ if vertices of type $\alpha$ build edges with exposure $r$ to vertices of type $\beta$. A shock $M$ such that $\E[W^{+,r,\alpha}\1\{A=\beta\}\1\{M=0\}]=0$ for all $(r,\alpha,\beta) \in \tilde{V}$ targets only lenders who are not borrowing themselves and it is not surprising that these shocks do not spread, i.e. the financial system is resilient with respect to these shocks. 
For all other shocks it turns out that they can be classified by checking for each $(r,\alpha,\beta) \in \tilde{V}$ if $\E[W^{+,r,\alpha}\1\{A=\beta\}\1\{M=0\}]$ is equal to $0$ or strictly greater than $0$ instead.
Consider for example a financial network consisting of banks of two types which are isolated of each other. Further, one of the two subnetworks shall be resilient, whereas the other one is non-resilient (in the sense of Definition \ref{def:non:resilience} (b)
). In order for the whole system to experience large damage, it is then necessary that $M$ does not only infect banks in the resilient subsystem but also in the non-resilient one. This explains why we have to differentiate between different choices for $M$ in the following to fully understand non-resilience in our model.

\begin{definition}[Non-Resilience]\label{def:non:resilience}
\begin{enumerate}[(a)]
\item Let $I\subseteq \tilde{V}:= \{ (r, \alpha , \beta ) \in V : \E[W^{+,r,\alpha}\1\{A=\beta\}]>0 \} $. A financial system is called \emph{non-resilient with respect to shocks on $I$} if there exists a constant $\Delta_I>0$ such that $n^{-1}\vert\mathcal{D}_n^M\vert \geq \Delta_I$ w.\,h.\,p.~for each shock $M$ with $\E[W^{+,r,\alpha}\1\{A=\beta\}\1\{M=0\}]>0$ for all $(r,\alpha,\beta)\in I$.
\item We call a financial system \emph{non-resilient} if it is non-resilient w.\,r.\,t.~shocks on 
some $I\subseteq \tilde{V}$.
\end{enumerate} 
\end{definition}
Clearly if a system is non-resilient with respect to shocks on $I$ it is also non-resilient with respects to shocks on $\tilde{I}$ for $I\subset \tilde{I}$. It follows that a system is non-resilient if and only if it is non-resilient w.\,r.\,t.~shocks on $\tilde{V}$.

\begin{remark}
Definition \ref{def:non:resilience} characterizes networks as non-resilient (with respect to shocks on $I$) if a lower bounded fraction of banks defaults. In many applications it might be desirable to distinguish between different parts of the global network and only call a system non-resilient if certain important parts experience lower bounded damage. This can be achieved by systemic importance values as introduced in Remark \ref{RemSysImp}: assign $s_i=0$ to banks $i$ in unimportant parts of the global network and replace $n^{-1}\vert\mathcal{D}_n^M\vert$ in Definitions \ref{def:resilience} and \ref{def:non:resilience} by $\mathcal{S}_n^M:=n^{-1}\sum_{i\in\mathcal{D}_n^M}s_i$. 
\end{remark}

Let us start by considering the special case that $M$ infects every part of the system (i.\,e.~a shock on $\tilde{V}$). This is the case for example if $M$ is independent of type $A$, vertex-weights $W^{\pm,r,\alpha}$ and capital $C$. We can then formulate a corollary of Theorem \ref{thm:non-resilience} stated below:
\begin{corollary}\label{cor:non:resilience:independent}
Consider a financial system described by a regular vertex sequence and  any ex post infection $M$ with $\E[W^{+,r,\alpha}\1\{A=\beta\}\1\{M=0\}]>0$ for all $(r,\alpha,\beta)\in \tilde{V}$. Then w.\,h.\,p.{} $n^{-1}\vert\mathcal{D}_n^M\vert \geq g(\bm{z}^*)$. If $\bm{z}^*\neq\bm{0}$ (i.\,e.~$S_0\neq\{\bm{0}\}$), then $g(\bm{z}^*)>0$ and the system is non-resilient.
\end{corollary}
By Theorem \ref{thm:resilience} and Corollary \ref{cor:non:resilience:independent} resilience of a financial system is hence completely characterized by $S_0$. See Figure \ref{fig:non:resilient} for an example where $S_0\neq\{\bm{0}\}$. In this example, we chose $R=2$, $T=1$, weights $W^{\pm,1}=W^{\pm,2}=3/2$ and capital $C=2$. 
The figure shows the jump of the smallest joint root from $\bm{0}$ to above $\bm{z}^*$ for any small shock (here $10\%$ of all banks). The discontinuity of the roof set of $f^2$ (orange line) shows that particularly the edges with exposure $r=2$ make the system non-resilient. 
\begin{figure}[t]
	\hfill\includegraphics[width=0.5\textwidth]{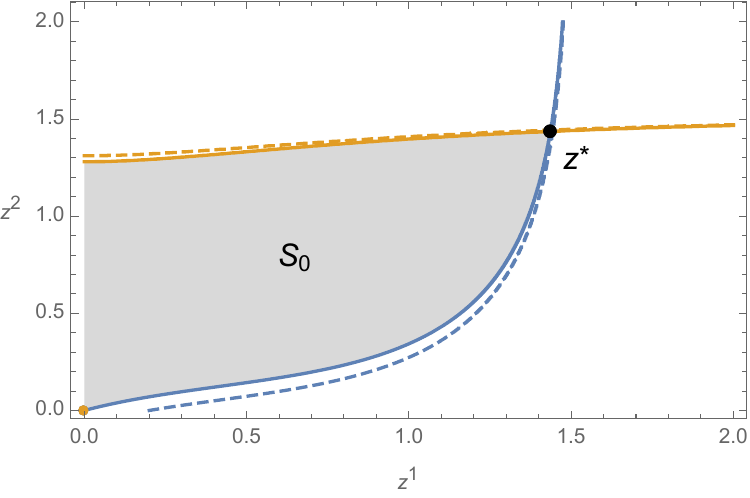}\hfill
		\caption{Plot of the root sets of the functions $f^1(z^1,z^2)$ (blue) and $f^2(z^1,z^2)$ (orange) for a financial system with 
		$S_0\neq\{\bm{0}\}$. Solid: the unshocked functions. Dashed: the shocked functions.
}\label{fig:non:resilient}
\end{figure}

We now aim to describe non-resilience with respect to shocks on $I\subsetneq \tilde{V}$. 
That is, we consider shocks $M$ such that $\E\left[W^{+,r,\alpha}\1\{A=\beta\}\1\{M=0\}\right]>0$ for $(r,\alpha,\beta)\in I$ but possibly $\E\left[W^{+,r,\alpha}\1\{A=\beta\}\1\{M=0\}\right]=0$ for 
$(r,\alpha,\beta)\in \tilde{V}\backslash I$. To this end, denote
\[ T(I) := \overline{\bigcap_{(r_1,\alpha_1,\beta_1)\in I}\left\{\bm{z}\in\R_{+,0}^V\,:\,f^{r_1,\alpha_1,\beta_1}(\bm{z})<0\right\} \cap \bigcap_{(r_2,\alpha_2,\beta_2)\in \tilde{V}\backslash I} \left\{\bm{z}\in\R_{+,0}^V\,:\,f^{r_2,\alpha_2,\beta_2}(\bm{z}) \leq 0\right\}} \]
and define $\bm{z}_0(I)$ 
by $z_0^{r,\alpha,\beta}(I) := \inf_{\bm{z}\in T(I)}z^{r,\alpha,\beta}$. Lemma \ref{lem:z0} shows that $\bm{z}_0(I)$ is the smallest joint root of 
$f^{r,\alpha,\beta}$, $(r,\alpha,\beta)\in V$, that is stable with respect to shocks in the $I$-coordinates. 
\begin{lemma}\label{lem:z0}
It holds 
$\bm{z}_0(I) \in S_0 \cap T(I)$ and it is thus a joint root of 
$f^{r,\alpha,\beta}$, $(r,\alpha,\beta)\in V$.
\end{lemma}

We can state a general theorem for non-resilience in terms of $\bm{z}_0(I)$, where $I$ denotes the set of coordinates impacted by 
$M$. 
\begin{theorem}[Non-Resilience Criterion]\label{thm:non-resilience}
Consider a financial system described by a regular vertex sequence and any ex post infection $M$ with $\E[W^{+,r,\alpha}\1\{A=\beta\}\1\{M=0\}]>0$ for all $(r,\alpha,\beta)\in I$, where $\emptyset\neq I\subseteq \tilde{V}$. Then w.\,h.\,p.~$n^{-1}\vert\mathcal{D}_n^M\vert \geq g(\bm{z}_0(I))$. If $\bm{z}_0(I)\neq\bm{0}$, then $g(\bm{z}_0(I))>0$ and the system is thus non-resilient with respect to shocks on $I$.
\end{theorem}
By Theorem \ref{thm:non-resilience} we derive for shocks on $\tilde{V}$ that for any $\epsilon>0$ w.\,h.\,p.~$n^{-1}\vert\mathcal{S}_n^M\vert \geq g(\bm{z}_0(\tilde{V}))-\epsilon$ while in Corollary \ref{cor:non:resilience:independent} we claimed $n^{-1}\vert\mathcal{S}_n^M\vert \geq g(\bm{z}^*)-\epsilon$ w.\,h.\,p. By the following lemma the two are in fact equivalent.
\begin{lemma}\label{lem:z0:equals:z*}
It holds $\bm{z}_0(\tilde{V})=\bm{z}^*$.
\end{lemma}
The identity $\bm{z}_0(I)=\bm{z}^*$ on the other hand does not necessarily imply $I=\tilde{V}$.

\section{Applications}\label{sec:applications}
The theory developed in the previous sections allows to investigate many interesting novel settings as compared to the present literature. In this section, we discuss some of them and highlight their implications. Further, we demonstrate the applicability of our asymptotic results also for finite networks of reasonable size by numerical simulations.

In the first example, we investigate the influence of a non-resilient subsystem in a global system. Unsurprisingly the global system turns out to be non-resilient as well and we can further show that even resilient network parts become non-resilient by their connections to the non-resilient subsystem, i.\,e.~every howsoever small infection that occurs only within the resilient part of the system finally spreads to a lower bounded fraction of the resilient subsystem. 

\begin{example}\label{ex:non-res:subsystem}
For simplicity assume $R=1$ and denote $z^{\alpha,\beta}:=z^{1,\alpha,\beta}$, $f^{\alpha,\beta}(\bm{z}):=f^{1,\alpha,\beta}(\bm{z})$ and $W^{\pm,\alpha}:=W^{\pm,1,\alpha}$ in the following. Consider then a $1$-type banking system, described by the random vector $(\tilde{W}^-,\tilde{W}^+,\tilde{C})$, where $\P(\tilde{W}^+>0)=1$ and $\P(\tilde{C}=0)=0$, and assume that it is non-resilient
. In the $1$-dimensional case this breaks down to the existence of $\tilde{z}_0>0$ such that
\[ \tilde{f}(z) := \E\left[\tilde{W}^+\P\left(\mathrm{Poi}\left(\tilde{W}^-z\right)\geq\tilde{C}\right)\right] - z \geq 0, \quad\text{for all }z\in[0,\tilde{z}_0]. \]

Now introduce a second (possibly resilient) subsystem to the network. That is, the system is now described by the random vector $(W^{\pm,1,1},W^{\pm,1,2},W^{\pm,2,1},W^{\pm,2,2},C,A)$, where $\P(C=0)=0$, $A\in\{1,2\}$ and $\alpha_i=1$ means that bank $i\in[n]$ is in the non-resilient subsystem, whereas $\alpha_i=2$ means that $i$ is part of the second subsystem. In order for the characteristics of the non-resilient subsystem to be preserved, we require that $W^{-,1}\vert_{A=1} \stackrel{d}{=} \tilde{W}^-$, $W^{+,1}\vert_{A=1} \stackrel{d}{=} \tilde{W}^+/\P(A=1)$ (to account for the changed number of banks; due to the multiplicative form in \eqref{eqn:edge:prob} it is sufficient to adjust either in- or out-weights by $\P(A=1)$) and $C\vert_{A=1}\stackrel{d}{=}\tilde{C}$. We derive that
\[ f^{1,1}(\bm{z}) = \E\left[W^{+,1}\P\left(\mathrm{Poi}\left(W^{-,1}z^{1,1}+W^{-,2}z^{1,2}\right)\geq C\right)\1\{A=1\}\right] - z^{1,1} \geq \tilde{f}\left(z^{1,1}\right) \geq 0 \]
for all $\bm{z}=(z^{1,1},z^{1,2},z^{2,1},z^{2,2})$ with $z^{1,1}\in[0,\tilde{z}_0]$ and in particular $z_0^{1,1}(I)\geq\tilde{z}_0>0$, where $I:=\{(1,1)\}$. 
An application of Theorem \ref{thm:non-resilience} then yields that the fraction of finally defaulted banks in the network is lower bounded by
\begin{align*}
g(\bm{z}_0(I)) &= \E\left[\P\left(\mathrm{Poi}\left(W^{-,1}z_0^{1,1}(I)+W^{-,2}z_0^{1,2}(I)\right)\geq C\right)\1\{A=1\}\right]\\
&\hspace{4.5cm}+ \E\left[\P\left(\mathrm{Poi}\left(W^{-,1}z_0^{2,1}(I)+W^{-,2}z_0^{2,2}(I)\right)\geq C\right)\1\{A=2\}\right]
\end{align*}
w.\,h.\,p.~for any ex post infection $M$ satisfying $\P(M=0, A=1)>0$ (i.\,e.~infecting some banks in the non-resilient subsystem). That is, if a small fraction of banks in the non-resilient subsystem defaults due to an external shock event, then this infection spreads to the whole system and the fraction of finally defaulted banks in the second subsystem is w.\,h.\,p.~lower bounded~by
\begin{equation}\label{eqn:lower:bound:second:subsystem}
\E\left[\left.\P\left(\mathrm{Poi}\left(W^{-,1}z_0^{2,1}(I)+W^{-,2}z_0^{2,2}(I)\right)\geq C\right)\,\right\vert\,A=2\right].
\end{equation}
In fact, if we assume that $W^{+,2}\vert_{A=1}>0$ almost surely and $\P(W^{-,1}>0,C<\infty,A=2)>0$ (that is, there are some banks in the second subsystem lending to banks in the non-resilient subsystem), then it must hold that
\[ z_0^{2,1}(I) \geq \E\left[W^{+,2}\P\left(\mathrm{Poi}\left(W^{-,1}z_0^{1,1}(I)\right)\geq C\right)\1\{A=1\}\right] >0 \]
and hence the lower bound \eqref{eqn:lower:bound:second:subsystem} is strictly positive. That is, every howsoever small infected fraction in the non-resilient subsystem spreads to a lower bounded fraction of finally defaulted banks in the second subsystem as well.

Now finally assume that $W^{+,1}\vert_{A=2}>0$ almost surely and $\P(W^{-,2}>0,C<\infty,A=1)>0$ (that is, there are some banks in the non-resilient subsystem lending to banks in the second subsystem). By considering the function
\[ f^{1,2}(\bm{z}) = \E\left[W^{+,1}\P\left(\mathrm{Poi}\left(W^{-,1}z^{2,1}+W^{-,2}z^{2,2}\right)\geq CM\right)\1\{A=2\}\right] - z^{1,2}, \]
we derive that $(\hat{z}^M)^{1,2}>0$ for any ex post infection $M$ such that $\P(M=0,A=2)>0$ (that is, infecting some banks in the second subsystem) and hence also $(\hat{z}^M)^{1,1}>0$ by the form of $f^{1,1}(\bm{z})$ (see above). By the same means as before, we hence conclude that in fact $\hat{\bm{z}}^M\geq\bm{z}_0(I)$ and so the lower bounds derived above still hold. In particular, this means that every howsoever small initial shock to the second (possibly resilient) subsystem causes the default of a lower bounded fraction of banks in the second subsystem. That is, by connecting to the non-resilient subsystem the a priori possibly resilient second subsystem becomes non-resilient as well.

\end{example}

To better understand the phenomenon in Example \ref{ex:non-res:subsystem}, we specify all parameters explicitly: Let
\begin{align*}
W^{-,1}\vert_{A=1} &= w_1, & W^{+,1}\vert_{A=1} &= 2w_1, & C\vert_{A=1} &= 1,\\
W^{-,2}\vert_{A=2} &= w_2, & W^{+,2}\vert_{A=2} &= 2w_2, & C\vert_{A=2} &= 2,
\end{align*}
for $w_1>1$ and $w_2>0$. The parameters $w_1$ and $w_2$ are then a measure for how strongly connected the respective subsystems are and it is easy to confirm that the type-1 subsystem is in fact non-resilient whereas the type-2 subsystem is resilient. Both subnetworks are Erd\"{o}s-R\'{e}nyi random graphs. In the first subnetwork every edge is contagious and because of $w_1>1$ there exists a giant component (a component of size $\lambda n$ for some $\lambda>0$). As soon as one vertex in the giant component defaults, the entire giant component defaults. The second subnetwork can be easily seen to be resilient by analysing the respective functional. Additionally, we assume that both subsystems have the same size and edges from subsystem $1$ to subsystem $2$ are as likely as edges from subsystem $2$ to subsystem $1$, i.e. $\P(A=1)=\P(A=2)=1/2$ and $W^{\pm,1}\vert_{A=2} = W^{\pm,2}\vert_{A=1} = w_3 >0$. The parameter $w_3$ is a measure for how strongly interconnected the different subsystems are. In particular, 
\begin{align*}
w_3\left(f^{1,1}(\bm{z})+z^{1,1}\right) &= 2w_1\left(f^{2,1}(\bm{z})+z^{2,1}\right), & w_3\left(f^{2,2}(\bm{z})+z^{2,2}\right) &= 2w_2\left(f^{1,2}(\bm{z})+z^{1,2}\right)
\end{align*}
and hence it must hold that $z_0^{2,1}(I) = (2w_1)^{-1}w_3\,z_0^{1,1}(I)$ resp.~$z_0^{1,2}(I) =(2w_2)^{-1}w_3\,z_0^{2,2}(I)$. The problem then reduces to $f^1(z^1,z^2) = 0$ and $f^2(z^1,z^2)=0$, where $z^1:=z^{1,1}$, $z^2:=z^{2,2}$ and
\begin{align*}
f^1(z^1,z^2) &:= w_1\left(1-e^{-w_1z^1-(2w_2)^{-1}w_3^2z^2}\right) - z^1,\\
f^2(z^1,z^2) &:= w_2\left(1-e^{-(2w_1)^{-1}w_3^2z^1-w_2z^2}\left(1+(2w_1)^{-1}w_3^2z^1+w_2z^2\right)\right) - z^2.
\end{align*}
Depending on the choice of the weights $w_i$, $i=1,2,3$, the system shows slightly different behaviour, as illustrated in Figures \ref{fig:ex1:1jointRoot}, \ref{fig:ex1:1jointRoot:hump} and \ref{fig:ex1:3jointRoots}. In all cases, 
$\bm{z}_0(I)=\bm{z}^*\neq\bm{0}$ which determines a strictly positive lower bound on the final default fraction as shown in Example \ref{ex:non-res:subsystem}. Further in all three cases the type $2$ subsystem is resilient. However, in \ref{fig:ex1:1jointRoot:hump} and \ref{fig:ex1:3jointRoots} the root sets for the function $f^2(z^1,z^2)$ shows a discontinuity even along the axis $z^1=0$. Let us elaborate on what this means. Along $z^1=0$ the function $f^2(0,z^2)$ describes contagion within the resilient subsystem $2$, not taking contagion between the subsystems into account. The discontinuity of set $S_0$ on the axis $z^1=0$ then means the following: For small shocks on subsystem $2$ contagion is being contained. As the shock gets larger, however, the gap in $S_0$ is overcome and the final fraction of defaulted type $2$ institutions experiences a jump above the upper end of $S_0$ on $z^1=0$. Considering the full picture again additionally to the small initial shock, subsystem $2$ experiences a shock by the positive fraction of defaulted type $1$ institutions. While in \ref{fig:ex1:3jointRoots}  their number is not sufficient to overcome the gap in $S_0$ and the final fraction of type $2$ institutions stays small (but positive), in \ref{fig:ex1:1jointRoot:hump} with larger parameter $w_3$ and hence denser connection between the subsystems, the shock on subsystem $2$ imposed by defaulted type 1 instututions is strong enough to overcome the gap and cause severe damage to subsystem $2$. The discontinuity phenomena has been observed in \cite{Janson2012} for the Erd\"{o}s-R\'{e}nyi random graph.

\begin{figure}[t]
    \hfill\subfigure[]{\includegraphics[width=0.32\textwidth]{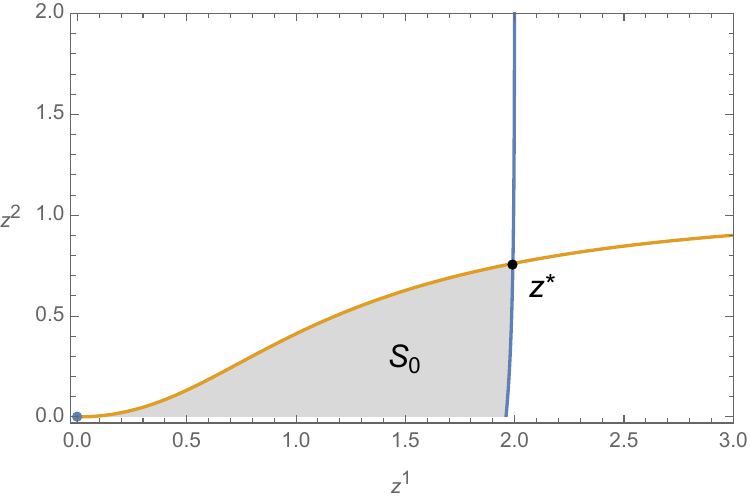}\label{fig:ex1:1jointRoot}}
    \hfill\subfigure[]{\includegraphics[width=0.32\textwidth]{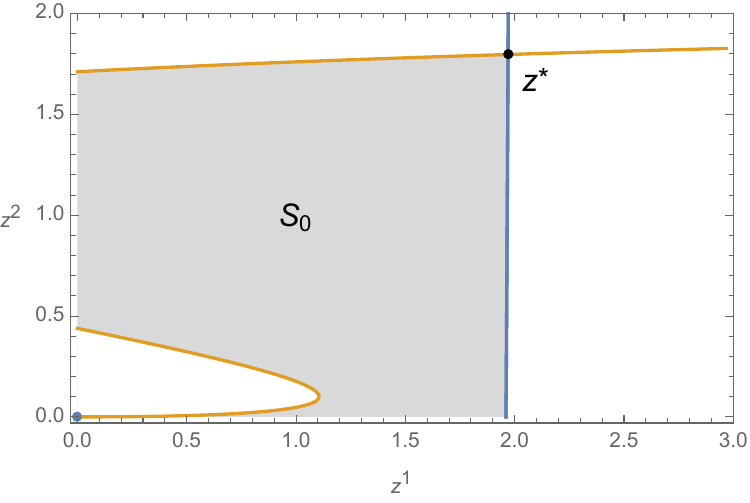}\label{fig:ex1:1jointRoot:hump}}
    \hfill\subfigure[]{\includegraphics[width=0.32\textwidth]{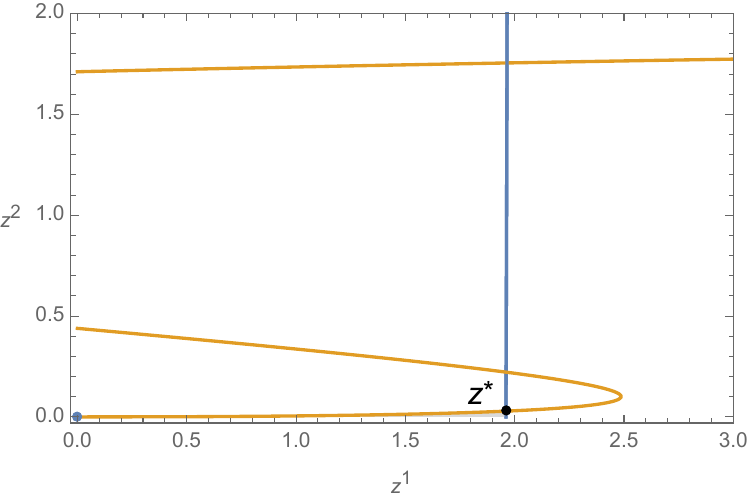}\label{fig:ex1:3jointRoots}}\hfill
\caption{Plot of the root sets of the functions $f^1(z^1,z^2)$ (blue) and $f^2(z^1,z^2)$ (orange) for the system with (a) $w_1=2$, $w_2=1$ and $w_3=2$, (b) $w_1=2$, $w_2=2$ and $w_3=3/4$ respectively (c) $w_1=2$, $w_2=2$ and $w_3=1/2$.
}\label{fig:Ex1}
\end{figure}

Since all the main results of this article and the derivations in Example \ref{ex:non-res:subsystem} are asymptotical, we demonstrate the applicability for finite networks numerically: For each of the scenarios (a)-(c) in Figure \ref{fig:Ex1} we performed $10^4$ simulations on networks of varying size $n\in\{100k\,:\,k\in[100]\}$ with $1\%$ initially defaulted banks. The outcomes are plotted in Figure \ref{fig:Convergence:Ex1} together with the theoretical asymptotic final default fraction (taking into account the initial default fraction of $1\%$). For case (a) except for 6 simulations at $n=100$ all results lie considerably close to the theoretical final fraction of $\approx87.98\%$ and their deviation becomes smaller the larger $n$ grows. For case (b) and $n<10^3$ some of the simulations ended in final default fractions around $55\%$. (Graphically these come from deviations of the hump (root set of $f^2$, orange) in Figure \ref{fig:ex1:1jointRoot:hump} such that it intersects with the root of $f^1$ (blue)). For all other simulations and especially for $n\geq10^3$, the simulation results clearly converge to the theoretical value of $\approx94.25\%$. For case (c), finally, some of the simulation outcomes for $n\leq500$ were close to $0$ and few around $\approx92.63\%$ (the value if one plugs in the largest of the three joint roots into $g$). The majority of the simulations (in particular for $n\geq4000$), however, resulted in final default fractions close to the theoretical value of $50.02\%$ and again deviations decrease as $n$ increases and we clearly confirm the theoretical convergence. Altogether we conclude that already for finite \mbox{networks of a few thousand vertices our asymptotic results are applicable with good accuracy.}

\begin{figure}[t]
\centering
    \hfill\subfigure[]{\includegraphics[width=0.49\textwidth]{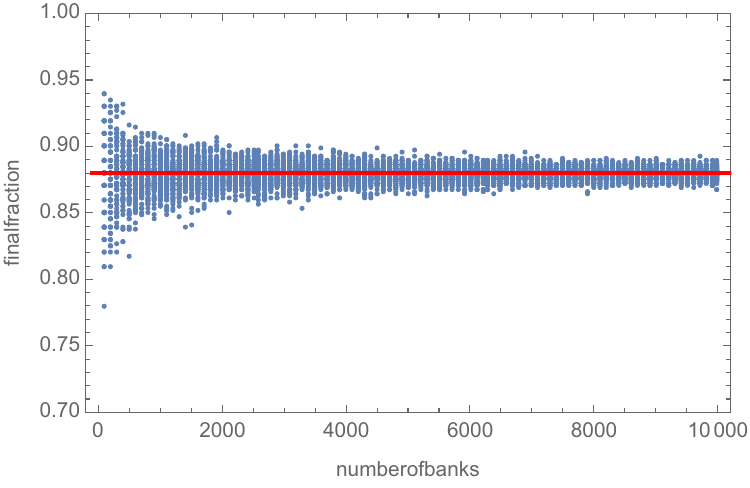}\label{fig:ex1:convergence:1jointRoot}}
    \hfill\subfigure[]{\includegraphics[width=0.49\textwidth]{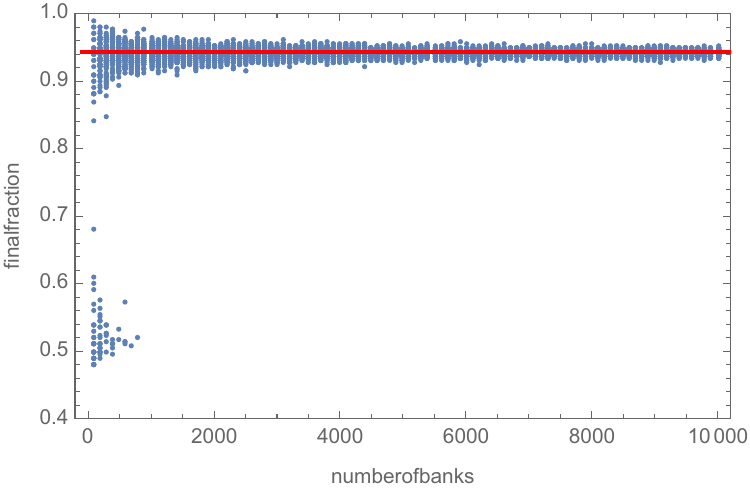}\label{fig:ex1:convergence:1jointRoot:hump}}
    \hfill\\
    \hfill\subfigure[]{\includegraphics[width=0.49\textwidth]{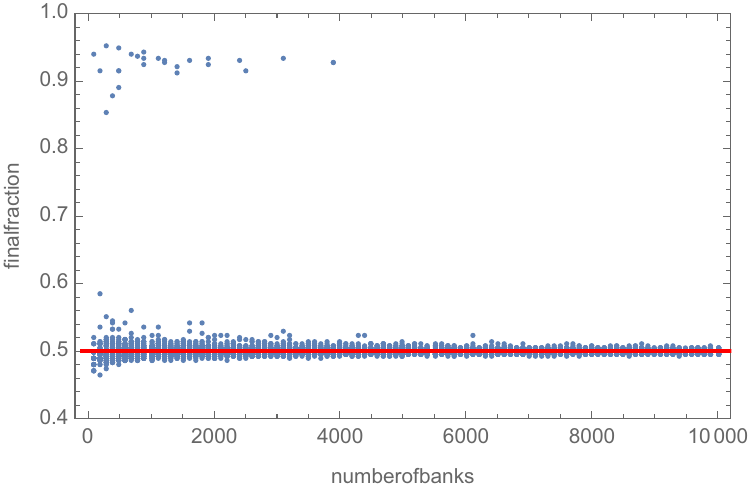}\label{fig:ex1:convergence:3jointRoots}}\hfill
\caption{Plot of the simulation results on networks of varying size (blue) and the theoretical asymptotic final default fraction (red) for the system with (a) $w_1=2$, $w_2=1$ and $w_3=2$, (b) $w_1=2$, $w_2=2$ and $w_3=3/4$ respectively (c) $w_1=2$, $w_2=2$ and $w_3=1/2$.
}\label{fig:Convergence:Ex1}
\end{figure}

In particular, what we learn from Example \ref{ex:non-res:subsystem} is that in order to ensure resilience of a particular subsystem, one needs to completely prohibit links to other non-resilient subsystems. It is, however, also possible that two subsystems which are resilient on their own form a non-resilient global system once connected to each other. It is therefore an interesting regulatory question how to ensure also resilience of a global system composed of various resilient subsystems. In general for our model the answer to this question is provided by Theorem \ref{thm:resilience}. However, in the following example we state a more intuitive criterion.

\begin{example}\label{ex:res:global:system}
Again, for simplicity assume that $R=1$. Consider a financial network consisting of $T$ subnetworks (types) which shall satisfy the following $1$-dimensional resilience conditions: For each $\epsilon>0$ there exists $z_\epsilon>0$ such that for all $z\in(0,z_\epsilon)$ and $\alpha\in[T]$ it holds
\begin{equation}\label{eqn:capital:requirements:subsystems}
\epsilon>\E\left[W^{+,\alpha}W^{-,\alpha}\P\left(\mathrm{Poi}\left(W^{-,\alpha}z\right)=C-1\right)\1\{A=\alpha\}\right].
\end{equation}
Note that this condition implies $\E[W^{+,\alpha}\P(\mathrm{Poi}(W^{-,\alpha}z)\geq C)\vert A=\alpha]<0$ for all $z$ small enough and hence indeed it implies resilience of the subsystem by Theorem \ref{thm:resilience}. In \cite{Detering2016} explicit capital requirements (i.\,e.~a formula for $C\vert_{A=\alpha}$ in dependence of $W^{-,\alpha}\vert_{A=\alpha}$) were derived for the case of Pareto distributed weights (which are typically observed in real networks) which ensure 
\eqref{eqn:capital:requirements:subsystems}. 

Now further assume that there exists a constant $K<\infty$ such that
\begin{equation}\label{eqn:bound:on:international:links}
W^{\pm,\beta}\vert_{A=\alpha} \leq KW^{\pm,\alpha}\vert_{A=\alpha}\quad\text{almost surely},
\end{equation}
for all $\alpha\neq\beta\in[T]$, i.\,e.~the tendency of institutions to develop links with institutions outside their subnetwork is bounded by a constant multiple of their tendency to develop links with institutions within their subnetwork. In particular this is the case if the external weights are bounded from above and the internal weights are bounded from below.

Replacing $W^{\pm,\beta}\vert_{A=\alpha}$ by $KW^{\pm,\alpha}\vert_{A=\alpha}$ only makes the system less resilient (if the weights increase, the number of links increases and hence the total exposure of each institution). Hence set $\tilde{W}^{\pm,\beta}\vert_{A=\alpha} = KW^{\pm,\alpha}\vert_{A=\alpha}$ for $\alpha\neq\beta$ and $\tilde{W}^{\pm,\alpha}\vert_{A=\alpha}=W^{\pm,\alpha}\vert_{A=\alpha}$. Now define $\bm{v}\in\R_+^{[T]\times[T]}$ by $v^{\alpha,\beta}=K^{\1\{\alpha\neq\beta\}}$, $\alpha,\beta\in[T]$. Then we derive that
\begin{align*}
&\E\Bigg[\tilde{W}^{+,\alpha}\Bigg(\sum_{\beta'\in[T]}v^{\alpha,\beta'}\tilde{W}^{-,\beta'}\Bigg)\P\Bigg(\mathrm{Poi}\Bigg(\sum_{\gamma\in[T]}\tilde{W}^{-,\gamma}z^{\alpha,\gamma}\Bigg)=C-1\Bigg)\1\{A=\alpha\}\Bigg]\\
&\hspace{0.5cm} = \E\Bigg[W^{+,\alpha}W^{-,\alpha}(1+K^2(T-1))\P\Bigg(\mathrm{Poi}\Bigg(W^{-,\alpha}\Bigg(z^{\alpha,\alpha}+K\sum_{\gamma\neq\alpha}z^{\alpha,\gamma}\Bigg)\Bigg)=C-1\Bigg)\1\{A=\alpha\}\Bigg]\\
&\hspace{0.5cm} < (1+K^2(T-1))\epsilon = v^{\alpha,\alpha}(1+K^2(T-1))\epsilon,
\end{align*}
for $z^{\alpha,\alpha}+K\sum_{\gamma\neq\alpha}z^{\alpha,\gamma}<z_\epsilon$, and
\begin{align*}
&\E\Bigg[\tilde{W}^{+,\alpha}\Bigg(\sum_{\beta'\in[T]}v^{\beta,\beta'}\tilde{W}^{-,\beta'}\Bigg)\P\Bigg(\mathrm{Poi}\Bigg(\sum_{\gamma\in[T]}\tilde{W}^{-,\gamma}z^{\beta,\gamma}\Bigg)=C-1\Bigg)\1\{A=\beta\}\Bigg]\\
&\hspace{0.24cm} = \E\Bigg[KW^{+,\beta}W^{-,\beta}(1+K^2(T-1))\P\Bigg(\mathrm{Poi}\Bigg(W^{-,\beta}\Bigg(z^{\beta,\beta}+K\sum_{\gamma\neq\beta}z^{\beta,\gamma}\Bigg)\Bigg)=C-1\Bigg)\1\{A=\beta\}\Bigg]\\
&\hspace{0.24cm} < K(1+K^2(T-1))\epsilon = v^{\alpha,\beta}(1+K^2(T-1))\epsilon,
\end{align*}
for $\alpha\neq\beta$ and $z^{\beta,\beta}+K\sum_{\gamma\neq\beta}z^{\beta,\gamma}<z_\epsilon$. If we now choose $\epsilon<(1+K^2(T-1))^{-1}$, then 
\[ \tilde{f}^{\alpha,\beta}(\delta\bm{v}):=\E\Bigg[\tilde{W}^{+,\alpha}\P\Bigg(\mathrm{Poi}\Bigg(\delta\sum_{\gamma\in[T]}\tilde{W}^{-,\gamma}v^{\beta,\gamma}\Bigg)\geq C\Bigg)\1\{A=\beta\}\Bigg] - \delta v^{\alpha,\beta} < 0 \]
for all $\delta>0$ small enough. It thus holds $\bm{z}^*\leq\lim_{\delta\to0+}\delta\bm{v}=\bm{0}$ and therefore $S_0=\{\bm{0}\}$. We can then apply Theorem \ref{thm:resilience} and obtain that the combined system is still resilient.

From a regulatory perspective it is hence enough to impose capital requirements 
described by \eqref{eqn:capital:requirements:subsystems} and to restrict links between different subsystems in the sense of \eqref{eqn:bound:on:international:links}.
\end{example}

In our first two examples we concentrated on the (non-)resilience of multi-type networks. For simplicity, we assumed that all edges carry the same exposure ($R=1$). 
Another interesting feature of our model, however, is that it allows for exposures that depend on the types of both the creditor and the debtor bank. The following example shows that this can indeed make a huge difference as compared to previous models in which exposures could only depend on the size/degree/type of the creditor bank. It considers two very similar financial systems whose only difference is that in one system exposures depend on both the creditor and debtor type and in the other system they depend on the type of the creditor bank only. As a consequence the first system will turn out to be non-resilient whereas the second one is resilient.

\begin{example}\label{ex:neighbor:dependent}
Consider a network of size $2\leq n\in\N$ in which (asymptotically) $p=1/3$ of all banks have type $1$ and the remaining $1-p=2/3$ banks have type $2$. That is, $T=2$. Further assume that for each pair of vertices $(i,j)\in[n]^2$ an edge from $i$ to $j$ shall be present with probability $4/n$. Edges between two banks of type $1$ shall carry exposure $2$ and all other edges exposure $1$. That is, if $\alpha_i=1$, then $w_i^{\pm,2,1}=w_i^{\pm,1,2}=2$ and $w_i^{\pm,1,1}=w_i^{\pm,2,2}=0$. If $\alpha_i=2$, then $w_i^{\pm,1,1}=w_i^{\pm,1,2}=2$ and $w_i^{\pm,2,1}=w_i^{\pm,2,2}=0$. Finally, all banks shall have capital $2$.

Then similarly as for Example \ref{ex:non-res:subsystem} the originally eight-dimensional system reduces to 
\begin{align*}
f^1(z^1,z^2) &= 2p\P\left(\mathrm{Poi}(2z^2)+2\mathrm{Poi}(2z^1)\geq 2\right) - z^1,\\
f^2(z^1,z^2) &= 2(1-p)\P\left(\mathrm{Poi}(2(z^1+z^2))\geq 2\right) - z^2.
\end{align*}
See Figure \ref{fig:neighbor:dependent} for an illustration of the root sets of $f^1$ and $f^2$. This figure already shows 
that $\bm{z}^*\neq\bm{0}$ and hence non-resilience by Theorem \ref{thm:non-resilience}. 
\begin{figure}[t]
    \hfill\subfigure[]{\includegraphics[width=0.45\textwidth]{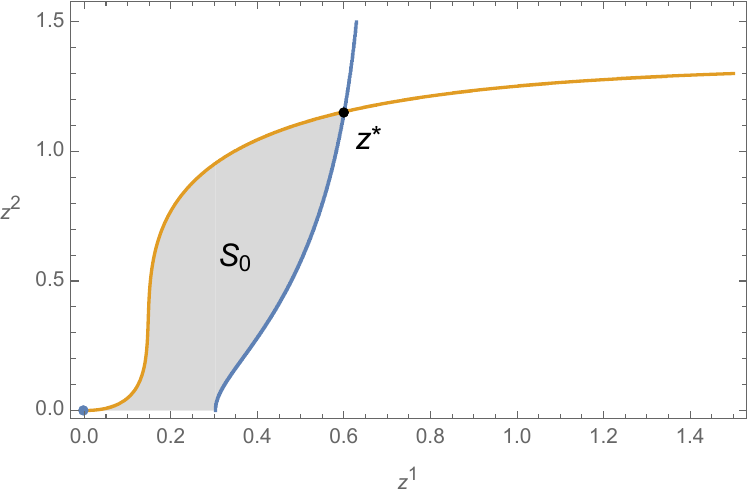}\label{fig:neighbor:dependent}}
    \hfill\subfigure[]{\includegraphics[width=0.45\textwidth]{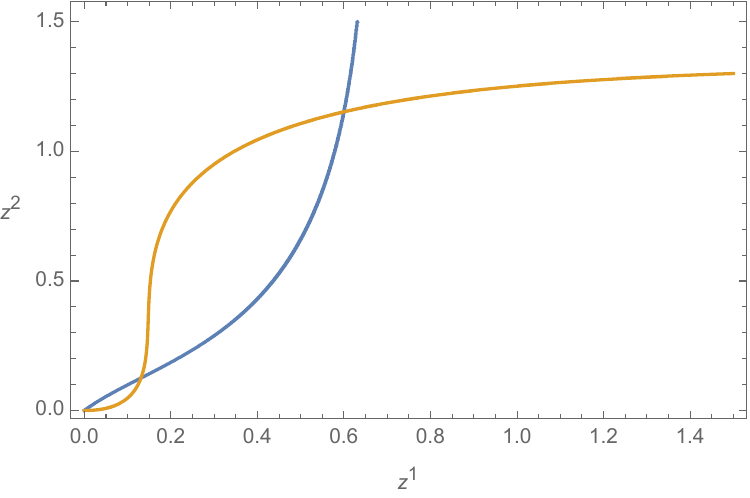}\label{fig:neighbor:independent}}\hfill
\caption{Plot of the root sets of the functions $f^1(z^1,z^2)$ (blue) and $f^2(z^1,z^2)$ (orange) for the system with (a) neighbor-dependent exposures respectively (b) neighbor-independent exposures.
}
\end{figure}
Also for $z^1,z^2\to0$, we can compute
\[ \frac{\partial f^1}{\partial z^1}(z^1,z^2) = 4p \P\left(\mathrm{Poi}(2z^2)+2\mathrm{Poi}(2z^1)\in\{0,1\}\right) - 1 \to 4p - 1 = \frac{1}{3} > 0, \]
which rigorously proves that the type-$1$ subnetwork and then also the whole system is non-resilient (cf.~Example \ref{ex:non-res:subsystem}). Numerically one derives that $\bm{z}^*\approx(0.601,1.153)$ and $g(\bm{z}^*)\approx 0.877$. In order to test this prediction, we performed $10^4$ simulations of financial networks of size $n=10^4$ with initial default probability $10^{-3}$. In only $5.32\%$ of the simulations, we observed a resilient nature in the sense that the simulated final default fraction was lower than $3\%$. All of the remaining simulations ended with a final default fraction within $[85.65\%,89.73\%]$ and are hence of a non-resilient nature. Averaging over the latter ones yields a mean final default fraction of $87.71\%$.


Now consider the following modified network: Instead of assigning exposure $2$ to all links between two type-$1$ banks and exposure $1$ to all other links, this time assign exposure $2$ with probability $p$ to any edge going to a type-$1$ bank (all other edges are assigned exposure $1$). That is, we keep the skeleton of the network but we redistribute the exposures in such a way that they do only depend on the creditor bank and not on the debtor bank. The total number of edges with exposure $2$ stays the same (note that in the first network the type-$1$ banks accounted to a fraction of $p$ of all the debtor banks of type-$1$ banks). This can be achieved by assigning the following new vertex-weights: $w_i^{+,1,1}=w_i^{+,2,1}=2$ for all $i\in[n]$. Further, if $\alpha_i=1$, then $w_i^{-,1,1}=w_i^{-,1,2}=2(1-p)$ and $w_i^{-,2,1}=w_i^{-,2,2}=2p$. All other vertex-weights shall stay the same. The new system then reduces to the following two functions, whose root sets are shown in Figure \ref{fig:neighbor:independent}:

\vspace*{-0.5cm}
\begin{align*}
f^1(z^1,z^2) &= 2p\P\left(\mathrm{Poi}(2(1-p)(z^1+z^2))+2\mathrm{Poi}(2p(z^1+z^2))\geq 2\right) - z^1,\\
f^2(z^1,z^2) &= 2(1-p)\P\left(\mathrm{Poi}(2(z^1+z^2))\geq 2\right) - z^2
\end{align*}

\vspace*{-0.1cm}
\noindent Figure \ref{fig:neighbor:independent} shows that the root set of $f^2$ is being shifted to the left, now starting off above the root set of $f^1$. One can hence already expect that $S_0=\{\bm{0}\}$ and the new system is resilient. Also more rigorously, as $z^1,z^2\to0$, we derive that

\vspace*{-0.6cm}
\begin{align*}
\frac{\partial f^1}{\partial z^1}(z^1,z^2) &= 4p(1-p)\P\left(\mathrm{Poi}(2(1-p)(z^1+z^2))+2\mathrm{Poi}(2p(z^1+z^2))=1\right)\\
&\hspace{3.15cm} + 4p^2\P\left(\mathrm{Poi}(2(1-p)(z^1+z^2))+2\mathrm{Poi}(2p(z^1+z^2))\in\{0,1\}\right) - 1\\
&\to 4p^2 - 1 = -\frac{5}{9},\\
\frac{\partial f^1}{\partial z^2}(z^1,z^2) &= 4p(1-p)\P\left(\mathrm{Poi}(2(1-p)(z^1+z^2))+2\mathrm{Poi}(2p(z^1+z^2))=1\right)\\
&\hspace{3.83cm} + 4p^2\P\left(\mathrm{Poi}(2(1-p)(z^1+z^2))+2\mathrm{Poi}(2p(z^1+z^2))\in\{0,1\}\right)\\
&\to 4p^2 = \frac{4}{9},\\
\frac{\partial f^2}{\partial z^1}(z^1,z^2) &= 4(1-p)\P\left(\mathrm{Poi}(2(z^1+z^2))=1\right) \to 0,\\
\frac{\partial f^2}{\partial z^2}(z^1,z^2) &= 4(1-p)\P\left(\mathrm{Poi}(2(z^1+z^2))=1\right) - 1 \to -1.
\end{align*}

\vspace*{-0.1cm}
\noindent The directional derivatives $D_{\bm{v}}f^1(\bm{0})$ and $D_{\bm{v}}f^2(\bm{0})$ thus exist for every $\bm{v}\in\R_+^V$. Choose then for example $\bm{v}=(v^1,v^2)=(1,1)$ such that $D_{\bm{v}}f^1(\bm{0})=-1/9$ and $D_{\bm{v}}f^2(\bm{0})=-1$. From Lemma \ref{lem:sufficient:criteria:z:star} we thus derive that $\bm{z}^*=\bm{0}$ and hence $S_0=\{\bm{0}\}$. This allows us to apply Theorem \ref{thm:resilience} and hence the modified system is indeed resilient. Again this can be validated numerically. On the same skeleton as for the previous simulation but with random exposures as described above the simulated final default fractions are now all within the interval $[0.11\%,0.63\%]$ with an average of $0.20\%$. The system is hence indeed of a resilient nature.
\end{example}

Although Example \ref{ex:neighbor:dependent} is too simple to model a real financial network, it still shows that counterparty-dependent exposures may have a significant impact on the stability of the system. In general it is also possible that they increase stability of the system, however.

\section{Proofs}\label{sec:proofs}
In this section we provide the proofs of 
our results in Sections \ref{sec:asymptotic:results} and \ref{sec:resilience}. Theorem \ref{thm:general:weights} will be proved in two steps. At this the underlying ideas are similar to \cite{Detering2015a} but at a considerable number of steps novel methods have to be used and we will particularly discuss them in detail.  We use the notation
\[ [\bm{a},\bm{b}]:=\bigcap_{(r,\alpha,\beta)\in V}\{\bm{z}\in\R^V\,:\,a^{r,\alpha,\beta}\leq z^{r,\alpha,\beta}\leq b^{r,\alpha,\beta}\} \]
for the cuboid spanned by the vectors $\bm{a}$ and $\bm{b}$ in $\R^V$ in the following. Further, let $\bm{\zeta}\in\R_{+,0}^V$ be defined by $\zeta^{r,\alpha,\beta}:=\E[W^{+,r,\alpha}\1\{A=\beta\}]$.

\subsection{Proofs of Lemmas \ref{lem:existence:hatz} and \ref{lem:sufficient:criteria:z:star}}\label{ssec:proofs:lemmas}
\begin{proof}[Proof of Lemma \ref{lem:existence:hatz}]
Existence of a smallest joint root $\hat{\bm{z}}\in[\bm{0},\bm{\zeta}]$ is ensured by the Knaster-Tarski fixed point theorem. We now construct a joint root $S_0\ni\bar{\bm{z}}\leq\hat{\bm{z}}$ which shows that $\hat{\bm{z}}=\bar{\bm{z}}\in S_0$, in particular:
It holds $f^{r,\alpha,\beta}(\hat{\bm{z}})=0$ for all $(r,\alpha,\beta)\in V$ and then $f^{r,\alpha,\beta}(\bm{z})\leq0$ for all $\hat{\bm{z}}\geq\bm{z}\in\R_{+,0}^V$ such that $z^{r,\alpha,\beta}=\hat{z}^{r,\alpha,\beta}$ 
by monotonicity of $f^{r,\alpha,\beta}$ from Lemma \ref{lem:properties:f}. 
Consider then the following sequence $(\bm{z}_n)_{n\in\N}\subset\R_{+,0}^V$:
\begin{itemize}
\item $\bm{z}_0=\bm{0}$
\item $\bm{z}_1=(z_1^{1,1,1},0,\ldots,0)$, where $z_1^{1,1,1}\geq0$ is the smallest possible value such that $f^{1,1,1}(\bm{z}_1)=0$. It is possible to find such $z_1^{1,1,1}$ by the intermediate value theorem since $f^{1,1,1}$ is continuous, $f^{1,1,1}(\bm{0})\geq 0$ and $f^{1,1,1}(\hat{z}^{1,1,1},0,\ldots,0)\leq0$. By Lemma \ref{lem:properties:f}, it then holds $f^{r,\alpha,\beta}(\bm{z}_1)\geq f^{r,\alpha,\beta}(\bm{0})\geq0$ for all $(1,1,1)\neq(r,\alpha,\beta)\in V$. In particular, $\bm{z}_1\in S_0$.
\item $\bm{z}_2=\bm{z}_1+(0,z_2^{1,1,2},0,\ldots,0)$, where $z_2^{1,1,2}\geq0$ is the smallest value such that $f^{1,1,2}(\bm{z}_2)=0$. Again it is possible to find such $z_2^{1,1,2}$ by the intermediate value theorem since $f^{1,1,2}$ is continuous, $f^{1,1,2}(\bm{z}_1)\geq0$ and $f^{1,1,2}(\bm{z}_1+(0,\hat{z}^{1,1,2},0,\ldots,0))\leq0$. Since $\bm{z}_1\in S_0$, by Lemma \ref{lem:properties:f} it then holds $f^{r,\alpha,\beta}(\bm{z}_2)\geq f^{r,\alpha,\beta}(\bm{z}_1)\geq0$ for all $(1,1,2)\neq(r,\alpha,\beta)\in V$. In particular, $\bm{z}_2\in S_0$.
\item $\bm{z}_i$, $i\in\{3,\ldots,RT^2\}$, are found analogously, changing only the corresponding coordinate.
\item $\bm{z}_{RT^2+1}=\bm{z}_{RT^2}+(z_{RT^2+1}^{1,1,1}-z_{RT^2}^{1,1,1},0,\ldots,0)$, where $z_{RT^2+1}^{1,1,1}\geq z_{RT^2}^{1,1,1}$ is the smallest value such that $f^{1,1,1}(\bm{z}_{RT^2+1})=0$, which is again possible by the intermediate value theorem. In particular, it still holds $z^{1,1,1}_{RT^2+1}\leq\hat{z}^{1,1,1}$. As before also $\bm{z}_{RT^2+1}\in S_0$.
\item Continue for $\bm{z}_i$, $i\geq RT^2+2$.
\end{itemize}
The sequence $(\bm{z}_n)_{n\in\N}$ constructed this way has the following properties: It is non-decreasing in each coordinate and $(\bm{z}_n)_{n\in\N}\subset S_0$. Further, it is bounded inside 
$[\bm{0},\hat{\bm{z}}]$. Hence by monotone convergence, each coordinate of $\bm{z}_n$ converges and so $\bar{\bm{z}}=\lim_{n\to\infty}\bm{z}_n$ exists. Now suppose there is $(r,\alpha,\beta)\in V$ such that $f^{r,\alpha,\beta}(\bar{\bm{z}})>0$. By continuity of $f^{r,\alpha,\beta}$ then also $f^{r,\alpha,\beta}(\bm{z}_n)>\epsilon$ for some $\epsilon>0$ and $n$ large enough. This, however, is in contradiction with the construction of the sequence $(\bm{z}_n)_{n\in\N}$ since $f^{r,\alpha,\beta}(\bm{z}_n)=0$ in every $RT^2$-th step. Hence $f^{r,\alpha,\beta}(\bar{\bm{z}})\leq 0$ for all $(r,\alpha,\beta)\in V$. Also $\bar{\bm{z}}\in S_0$, however, since this is a closed set
. Hence $f^{r,\alpha,\beta}(\bar{\bm{z}})\geq0$ for all $(r,\alpha,\beta)\in V$ and altogether this shows that $\bar{\bm{z}}$ is a joint root of all functions $f^{r,\alpha,\beta}$, $(r,\alpha,\beta)\in V$.


Now turn to the proof that $\bm{z}^*\in S_0$ and it is a joint root of all functions $f^{r,\alpha,\beta}$, $(r,\alpha,\beta)\in V$: First define the following sets for each $\epsilon>0$:
\[ S(\epsilon) := \bigcap_{(r,\alpha,\beta)\in V}\{\bm{z}\in\R_{+,0}^V\,:\,f^{r,\alpha,\beta}(\bm{z})\geq -\epsilon\}
\]
Further denote by $S_0(\epsilon)$ 
the connected component 
of $\bm{0}$ in $S(\epsilon)$
. By the same procedure as for $\hat{\bm{z}}$ above, we now derive that there exists a smallest (componentwise) point $\hat{\bm{z}}(\epsilon)\in S_0(\epsilon)$ such that $f^{r,\alpha,\beta}(\hat{\bm{z}}(\epsilon))=-\epsilon$ for all $(r,\alpha,\beta)\in V$. Clearly, $\hat{\bm{z}}(\epsilon)$ is non-decreasing in $\epsilon$ (componentwise) and hence 
$\tilde{\bm{z}}:=\lim_{\epsilon\to0+}\hat{\bm{z}}(\epsilon)$ exists (we will show 
that $\tilde{\bm{z}}=\bm{z}^*$ in fact).

Now by monotonicity of $S_0(\epsilon)$, we derive that $\hat{\bm{z}}(\delta)\in S_0(\delta)\subseteq S_0(\epsilon)$ for all $\delta\leq \epsilon$. Since $S_0(\epsilon)$ is a closed set, it must thus hold that also $\tilde{\bm{z}}=\lim_{\delta\to0+}\hat{\bm{z}}(\delta)\in S_0(\epsilon)$ for all $\epsilon>0$ and in particular, $\tilde{\bm{z}}\in\bigcap_{\epsilon>0}S_0(\epsilon)$. Further, by continuity of $f^{r,\alpha,\beta}$, $(r,\alpha,\beta)\in V$, we derive that $\bigcap_{\epsilon>0}S_0(\epsilon)\subseteq\bigcap_{\epsilon>0}S(\epsilon)\subseteq S$. Moreover, $\bigcap_{\epsilon>0}S_0(\epsilon)$ is the intersection of a chain of connected, compact sets in the Hausdorff space $\R^V$ and it is hence a connected, compact set itself. Since it further contains $\bm{0}$, we can then conclude that $\bigcap_{\epsilon>0}S_0(\epsilon)\subseteq S_0$ and thus $\tilde{\bm{z}}\in S_0$.

We now want to show that $\bm{z}\leq\tilde{\bm{z}}$ componentwise for arbitrary $\bm{z}\in S_0$. 
This clearly proves $\tilde{\bm{z}}=\bm{z}^*$. 
It thus suffices to show $S_0\subset[\bm{0},\hat{\bm{z}}(\epsilon)]$. Then $\bm{z}\leq\hat{\bm{z}}(\epsilon)$ and $\bm{z}\leq\lim_{\epsilon\to0+}\hat{\bm{z}}(\epsilon)=\tilde{\bm{z}}$. Hence assume that $S_0\not\subset[\bm{0},\hat{\bm{z}}(\epsilon)]$. By connectedness of $S_0$ we find $\bar{\bm{z}}\in S_0$ with $\bar{z}^{r,\alpha,\beta}\leq\hat{z}^{r,\alpha,\beta}(\epsilon)$ for all $(r,\alpha,\beta)\in V$ and equality for at least one coordinate. W.\,l.\,o.\,g.~let this coordinate be $(1,1,1)$. By monotonicity of $f^{1,1,1}$ with respect to $z^{r,\alpha,\beta}$ for every $(r,\alpha,\beta)\neq(1,1,1)$, we thus derive that
\[ f^{1,1,1}(\bar{\bm{z}})\leq f^{1,1,1}(\hat{\bm{z}}(\epsilon)) = -\epsilon. \]
However, we also assumed that $\bar{\bm{z}}\in S_0$ and hence $f^{1,1,1}(\bar{\bm{z}})\geq0$, a contradiction.

Finally, we obtain that $f^{r,\alpha,\beta}(\bm{z}^*)=f^{r,\alpha,\beta}(\tilde{\bm{z}}) = \lim_{\epsilon\to0+}f^{r,\alpha,\beta}(\hat{\bm{z}}(\epsilon)) = \lim_{\epsilon\to0+}(-\epsilon)=0$, by continuity of $f^{r,\alpha,\beta}$
. Hence $\bm{z}^*$ is in fact a joint root of all the functions $f^{r,\alpha,\beta}$, $(r,\alpha,\beta)\in V$.
\end{proof}
\begin{proof}[Proof of Lemma \ref{lem:sufficient:criteria:z:star}]
Note that it is sufficient to construct a sequence $(\bm{z}_n)_{n\in\N}\subset\R_+^V$ such that $\lim_{n\to\infty}\bm{z}_n=\bar{\bm{z}}$ and $f^{r,\alpha,\beta}(\bm{z}_n)<0$ for all $(r,\alpha,\beta)\in V$, $n\in\N$. By monotonicity of $f^{r,\alpha\beta}$ from Lemma \ref{lem:properties:f} it then follows that $\bm{z}^*\leq\bm{z}_n$ 
and hence $\bm{z}^*\leq\bar{\bm{z}}$. If condition 1.~is satisfied, we get $\lim_{n\to\infty}nf^{r,\alpha,\beta}\left(\bar{\bm{z}}+n^{-1}\bm{v}\right) = D_{\bm{v}}f^{r,\alpha,\beta}(\bar{\bm{z}}) < 0$ and we can hence choose $\bm{z}_n:=\bar{\bm{z}}+n^{-1}\bm{v}$.

If condition 2.~is satisfied, note that by Fubini's theorem for $n>\Delta^{-1}$ we derive
\begin{align*}
f^{r,\alpha,\beta}\bigg(\bar{\bm{z}}+\frac{1}{n}\bm{v}\bigg) &= \int_0^\frac{1}{n} \Bigg(-v^{r,\alpha,\beta} + \sum_{r'\in[R]}\E\Bigg[W^{+,r,\alpha}\Bigg(\sum_{\beta'\in[T]}v^{r',\beta,\beta'}W^{-,r',\beta'}\Bigg) \1\{A=\beta\}\\
&\hspace*{0.345cm}\times\P\Bigg(\sum_{s\in[R]}s\mathrm{Poi}\Bigg(\sum_{\gamma\in[T]}W^{-,s,\gamma}\Big(\bar{z}^{s,\beta,\gamma}+\delta v^{s,\beta,\gamma}\Big)\Bigg)\hspace*{-0.75ex}\in\hspace*{-0.5ex}\{C-r',\ldots,C-1\}\Bigg)\Bigg]\Bigg)\dd\delta\\
&\leq -n^{-1}(1-\kappa)v^{r,\alpha,\beta}
\end{align*}
and 
$\limsup_{n\to\infty}nf^{r,\alpha,\beta}\left(\bar{\bm{z}}+n^{-1}\bm{v}\right) \leq (1-\kappa)v^{r,\alpha,\beta} < 0$. 
Thus choose $\bm{z}_n:=\bar{\bm{z}}+n^{-1}\bm{v}$ again.
\end{proof}

\subsection{Proof of the Main Result for Finitary Weights}\label{ssec:proof:main:finitary}
In this section, we consider the special case that the vertex weights $w_i^{r,\alpha,\beta}$ and capitals $c_i$ can only take values in a finite set. To formalize this, consider the following definition:

\begin{definition}[Finitary regular vertex sequence]
A regular vertex sequence denoted by $(\bm{w}^{-,r,\alpha},\bm{w}^{+,r,\alpha},\bm{c},\bm{\alpha})$ is called finitary if there exist $J\in\N$ and $(\tilde{w}_j^{-,1,1},\ldots,\tilde{w}_j^{+,R,T})\in\R_{+,0}^{2RT}$, $j\in[J]$, as well as $c_\text{max}\in\N_0$ such that for all $n\in\N$ and $i\in[n]$, there exists $j=j(n,i)\in[J]$ such that $w_i^{\pm,r,\alpha}(n)=\tilde{w}_j^{\pm,r,\alpha}$ and $c_i(n)\in[c_\text{max}]\cup\{0,\infty\}$.
\end{definition}

That is, in a finitary system there is a partition of the set of all institutions into $TJ(c_\text{max}+2)$ 
sets. In particular, in this case all weights $w_i^{\pm,r,\alpha}$ are bounded from above by some constant $\overline{w}\in\R_+$ and hence by Dominated Convergence we can compute the partial derivatives of $f^{r,\alpha,\beta}$:
\begin{align*}
\frac{\partial f^{r,\alpha,\beta}}{\partial z^{r',\alpha',\beta'}}(\hat{\bm{z}}) &= -\delta_{r,r'}\delta_{\alpha,\alpha'}\delta_{\beta,\beta'} + \delta_{\beta,\alpha'}\E\Bigg[ W^{+,r,\alpha}W^{-,r',\beta'} \1\{A=\beta\}\\
&\hspace{3cm}\times\P\Bigg(\sum_{s\in[R]}s\mathrm{Poi}\Bigg(\sum_{\gamma\in[T]}W^{-,s,\gamma}z^{s,\beta,\gamma}\Bigg)\in\{C-r',\ldots,C-1\}\Bigg)\Bigg],
\end{align*}
where $\delta_{a,b}:=\1\{a=b\}$. Hence for any vector $\bm{v}\in\R^V$, the directional derivative of $f^{r,\alpha,\beta}$ in direction $\bm{v}$ is given by the following continuous expression:
\begin{align*}
D_{\bm{v}} f^{r,\alpha,\beta}(\hat{\bm{z}})&= 
- v^{r,\alpha,\beta} + \sum_{r'\in[R]} \E\Bigg[W^{+,r,\alpha}\Bigg(\sum_{\beta'\in[T]}v^{r',\beta,\beta'}W^{-,r',\beta'}\Bigg)\1\{A=\beta\}\\
&\hspace{3.6cm}\times\P\Bigg(\sum_{s\in[R]}s\mathrm{Poi}\Bigg(\sum_{\gamma\in[T]}W^{-,s,\gamma}z^{s,\beta,\gamma}\Bigg)\in\{C-r',\ldots,C-1\}\Bigg)\Bigg]
\end{align*}
We can then prove the following asymptotic results for the final default fraction in the network. 
\begin{proposition}\label{prop:finitary:weights}
Consider a financial system described by a finitary regular vertex sequence and let $\hat{\bm{z}}$ be the smallest joint root of the functions $\{f^{r,\alpha,\beta}\}_{(r,\alpha,\beta)\in V}$. Then it holds that $n^{-1}\vert\mathcal{D}_n\vert \geq g(\hat{\bm{z}}) + o_p(1)$. If additionally there exists 
$\bm{v}\in\R_+^V$ such that $D_{\bm{v}}f^{r,\alpha,\beta}(\hat{\bm{z}})<0$ for all $(r,\alpha,\beta)\in V$, then $n^{-1}\vert\mathcal{D}_n\vert = g(\hat{\bm{z}}) + o_p(1)$.
\end{proposition}

See Figure \ref{fig:one:joint:root} for an example where such $\bm{v}\in\R_+^V$ exists respectively Figure \ref{fig:two:joint:roots} for an example where it does not. Theorem \ref{thm:finitary:weights} below will analyze systems of the latter type as well.

\begin{proof}
We begin by proving the lower bound: As in \cite{Detering2015a} and \cite{Detering2016} we switch to a sequential default contagion process. The idea is to collect defaulted institutions and instead of exposing them all at once (as in \eqref{eqn:default:contagion}), only select one defaulted institution uniformly at random in each round $t\geq1$ and expose it to its neighbors (draw edges). Using the finitary assumption, it is then sufficient to keep track of the following sets and quantities during the default process:

\vspace*{-0.6cm}
\begin{align*}
U^\alpha(t) &:= \left\{i\in[n]\,:\,\alpha_i=\alpha\text{ and }i\text{ is defaulted but unexposed at time }t\right\},\\
S_{j,m,l}^\alpha(t) &:= \left\{ i\in[n]\,:\,\alpha_i=\alpha, w_i^{\pm,r,\beta}=\tilde{w}_j^{\pm,r,\beta},c_i=m\text{ and }i\text{ has total exposure of }l\text{ at time }t\right\},
\end{align*}

\vspace*{-1.0cm}
\begin{align*}
u^\alpha(t) &:= \left\vert U^\alpha(t)\right\vert, & c_{j,m,l}^\alpha(t) &:= \left\vert S_{j,m,l}^\alpha(t)\right\vert, & 
w^{r,\alpha,\beta}(t) &:= \sum_{i\in U^\beta(t)}w_i^{+,r,\alpha}.
\end{align*}

\vspace*{-0.35cm}
\noindent Let $h(t):=(u^\alpha(t),c_{j,m,l}^\alpha(t),
w^{r,\alpha,\beta}(t))$ and $H(t)=(h(s))_{s\leq t}$. Then for $n$ large enough such that all $p_{i,j}^r<R^{-1}$ (possible by finitary weights), the expected evolution of the system at 
time $t$ is 

\vspace*{-0.6cm}
\begin{align*}
&\E\left[\left.c_{j,m,l}^\alpha(t+1)-c_{j,m,l}^\alpha(t)\,\right\vert\, H(t)\right]\\
&\hspace{1.5cm}= \frac{1}{\sum_{\beta\in[T]}u^\beta(t)}\sum_{\beta\in[T]}\sum_{v\in U^\beta(t)} \sum_{r\in[R]}\left(\sum_{i\in S_{j,m,l-r}^\alpha(t)}\frac{w_v^{+,r,\alpha}w_i^{-,r,\beta}}{n} - \sum_{i\in S_{j,m,l}^\alpha(t)}\frac{w_v^{+,r,\alpha}w_i^{-,r,\beta}}{n}\right)\\
&\hspace{1.5cm}= \sum_{r\in[R]}\frac{\sum_{\beta\in[T]}w^{r,\alpha,\beta}(t) \tilde{w}_j^{-,r,\beta}}{\sum_{\beta\in[T]}u^\beta(t)}\frac{c_{j,m,l-r}^\alpha(t)-c_{j,m,l}^\alpha(t)}{n}.
\end{align*}

\vspace*{-0.7cm}
\begin{align*}
\E\left[\left.u^\alpha(t+1)-u^\alpha(t)\,\right\vert\, H(t)\right] &= -\frac{u^\alpha(t)}{\sum_{\beta\in[T]}u^\beta(t)} + \sum_{j,m} \sum_{r\in[R]}\frac{\sum_{\beta\in[T]}w^{r,\alpha,\beta}(t) \tilde{w}_j^{-,r,\beta}}{\sum_{\beta\in[T]}u^\beta(t)}\sum_{l=m-r}^{m-1}\frac{c_{j,m,l}^\alpha(t)}{n},
\end{align*}

\vspace*{-0.7cm}
\begin{align*}
&\E\left[\left.w^{r,\alpha,\beta}(t+1)-w^{r,\alpha,\beta}(t)\,\right\vert\, H(t)\right]\\
&\hspace{3.5cm}= -\frac{w^{r,\alpha,\beta}(t)}{\sum_{\gamma\in[T]}u^\gamma(t)} + \sum_{j,m} \tilde{w}_j^{+,r,\alpha} \sum_{s\in[R]}\frac{\sum_{\gamma\in[T]}w^{s,\beta,\gamma}(t)\tilde{w}_j^{-,s,\gamma}}{\sum_{\gamma\in[T]}u^\gamma(t)}\sum_{l=m-s}^{m-1}\frac{c_{j,m,l}^\beta(t)}{n}.
\end{align*}
The expressions on the right-hand side are all Lipschitz functions of $u^\alpha(t),w^{r,\alpha,\beta}(t),c_{j,m,l}^\alpha(t)$ as long as $\sum_{\beta\in[T]}u^\beta(t)$ is bounded away from zero. All the remaining conditions in Wormald’s theorem \cite{wormald1995} can be checked by similar means as in \cite{Detering2015a}. We can thus uniformly approximate 
\begin{align}
n^{-1}c_{j,m,l}^\alpha(t) &= \gamma_{j,m,l}^\alpha(n^{-1}t) + o_p(1),\label{eqn:approx:gamma}\\
n^{-1}u^\alpha(t) &= \nu^\alpha(n^{-1}t) + o_p(1),\label{eqn:approx:nu}\\
n^{-1}w^{r,\alpha,\beta}(t) &= \mu^{r,\alpha,\beta}(n^{-1}t) + o_p(1),\label{eqn:approx:mu}
\end{align}
where the functions $\gamma_{j,m,l}^\alpha(\tau)$, 
$\nu^\alpha(\tau)$ and $\mu^{r,\alpha,\beta}(\tau)$ are defined as the unique solution of
\begin{align}
\frac{\dd}{\dd \tau}\gamma_{j,m,l}^\alpha(\tau) &= \sum_{r\in[R]}\frac{\sum_{\beta\in[T]}\mu^{r,\alpha,\beta}(\tau)\tilde{w}_j^{-,r,\beta}}{\sum_{\beta\in[T]}\nu^\beta(\tau)}\left(\gamma_{j,m,l-r}^\alpha(\tau)-\gamma_{j,m,l}^\alpha(\tau)\right),\label{eqn:diff:gamma}\\
\frac{\dd}{\dd \tau} \nu^\alpha(\tau) &= -\frac{\nu^\alpha(\tau)}{\sum_{\beta\in[T]}\nu^\beta(\tau)}+\sum_{j,m}\sum_{r\in[R]}\frac{\sum_{\beta\in[T]}\mu^{r,\alpha,\beta}(\tau)\tilde{w}_j^{-,r,\beta}}{\sum_{\beta\in[T]}\nu^\beta(\tau)}\sum_{l=m-r}^{m-1}\gamma_{j,m,l}^\alpha(\tau),\label{eqn:diff:nu}\\
\frac{\dd}{\dd \tau}\mu^{r,\alpha,\beta}(\tau) &= -\frac{\mu^{r,\alpha,\beta}(\tau)}{\sum_{\gamma\in[T]}\nu^\gamma(\tau)}+\sum_{j,m}\tilde{w}_j^{+,r,\alpha}\sum_{s\in[R]}\frac{\sum_{\gamma\in[T]}\mu^{s,\alpha,\gamma}(\tau)\tilde{w}_j^{-,s,\gamma}}{\sum_{\gamma\in[T]}\nu^\gamma(\tau)}\sum_{l=m-r}^{m-1}\gamma_{j,m,l}^\alpha(\tau).\label{eqn:diff:mu}
\end{align}
Approximations \eqref{eqn:approx:gamma}-\eqref{eqn:approx:mu} hold uniformly for $t/n<\hat{\tau} := \inf\{\tau\in\R_{+,0}\,:\,\sum_{\beta\in[T]}\nu^\beta(\tau)=0\}$. For $z^{r,\alpha,\beta}(\tau) := \int_0^\tau\mu^{r,\alpha,\beta}(s)/\sum_{\gamma\in[T]}\nu^\gamma(s)\dd s$, an implicit solution of \eqref{eqn:diff:gamma}-\eqref{eqn:diff:mu} is given by
\begin{align*}
\gamma_{j,m,l}^\alpha(\tau) &= \P(W^{\pm,r,\alpha}=\tilde{w}_j^{\pm,r,\alpha}, C=m, A=\alpha) \P\Bigg(\sum_{s\in[R]}s\mathrm{Poi}\Bigg(\sum_{\beta\in[T]}\tilde{w}_j^{-,s,\beta}z^{s,\alpha,\beta}(\tau)\Bigg)=l\Bigg),\\
\nu^\alpha(\tau) &= \E\Bigg[\P\Bigg(\sum_{s\in[R]}s\mathrm{Poi}\Bigg(\sum_{\beta\in[T]}W^{-,s,\beta}z^{s,\alpha,\beta}(\tau)\Bigg)\geq C\Bigg)\1\{A=\alpha\}\Bigg] - \int_0^\tau\frac{\nu^\alpha(s)}{\sum_{\beta\in[T]}\nu^\beta(s)},
\end{align*}
\begin{align*}
\mu^{r,\alpha,\beta}(\tau) &= \E\Bigg[W^{+,r,\alpha}\P\Bigg(\sum_{s\in[R]}s\mathrm{Poi}\Bigg(\sum_{\gamma\in[T]}W^{-,s,\gamma}z^{s,\beta,\gamma}(\tau)\Bigg)\geq C\Bigg)\1\{A=\beta\}\Bigg] - z^{r,\alpha,\beta}(\tau).
\end{align*}


\noindent In particular, note that $\sum_{\alpha\in[T]}\nu^\alpha(\tau) = g(\bm{z}(\tau)) - \tau$ and $\mu^{r,\alpha,\beta}(\tau) = f^{r,\alpha,\beta}(\bm{z}(\tau))$. Thus for $\tau<\hat{\tau}$, it holds $f^{r,\alpha,\beta}(\bm{z}(\tau)) = n^{-1}w^{r,\alpha,\beta}(\lfloor\tau n\rfloor) + o_p(1) \geq 0 + o_p(1)$ 
and by letting $n\to\infty$, it follows $f^{r,\alpha,\beta}(\bm{z}(\tau))\geq0$. 
By continuity of $\bm{z}(\tau)$ and $\bm{z}(0)=0$, hence $\bm{z}(\tau)\in S_0$. 
Further,
\[ f^{r,\alpha,\beta}(\bm{z}(\tau)) = \mu^{r,\alpha,\beta}(\tau) = n^{-1}w^{r,\alpha,\beta}(\lfloor\tau n\rfloor) + o_p(1) \leq n^{-1} \overline{w} u^\beta(\lfloor\tau n\rfloor) + o_p(1) = \overline{w}\nu^\beta(\tau) + o_p(1) \]
and as $n\to\infty$, $f^{r,\alpha,\beta}(\bm{z}(\tau)) \leq \overline{w} \nu^\beta(\tau)$. 

As $\tau\to\hat{\tau}$, $\sum_{\beta\in[T]}\nu^\beta(\tau)\to0$ and hence by continuity of $\bm{z}(\tau)$, $f^{r,\alpha,\beta}(\bm{z}(\hat{\tau}))=0$ for all $(r,\alpha,\beta)\in V$. Again by continuity of $\bm{z}(\tau)$ and closedness of $S_0$, we then conclude that $\bm{z}(\hat{\tau})\geq\hat{\bm{z}}$. 
In particular,
\begin{equation}\label{eqn:tau:hat}
g(\hat{\bm{z}}) \leq g(\bm{z}(\hat{\tau})) = \lim_{\tau\to\hat{\tau}}g(\bm{z}(\tau)) =  \lim_{\tau\to\hat{\tau}} \sum_{\alpha\in[T]}\nu^\alpha(\tau)+\tau = \hat{\tau}.
\end{equation}
Hence, we are left with showing that $\hat{t}/n\geq\hat{\tau}+o_p(1)$ in order to prove the first part of the theorem, where $\hat{t}$ denotes the first time that $\sum_{\alpha\in[T]}u^\alpha(t)=0$. Define $X_n:=(\lfloor\hat{\tau}n\rfloor\wedge\hat{t})/n-\hat{\tau}$. Then $\hat{t}/n\geq\hat{\tau}+X_n$. Further, for $\epsilon>0$ and $n$ large enough such that $\hat{\tau}-\lfloor\hat{\tau} n\rfloor/n\leq\epsilon$, we obtain
\[ \P(\vert X_n\vert>\epsilon) = \P(\hat{\tau}-\hat{t}/n>\epsilon) \leq \P\Bigg(\sum_{\alpha\in[T]}\nu^\alpha(\hat{t}/n)>\frac{1}{2}\min_{\tau\in[0,\hat{\tau}-\epsilon]}\sum_{\alpha\in[T]}\nu^\alpha(\tau) , \hat{t}/n<\hat{\tau}\Bigg), \]
using continuity of $\sum_{\alpha\in[T]}\nu^\alpha(\tau)$. Let now $(Y_n)_{n\in\N}$ such that $\left\vert \sum_{\alpha\in[T]}u^\alpha(t)/n-\nu^\alpha(t/n)\right\vert\leq Y_n$ and $Y_n=o_p(1)$ (existence of $(Y_n)_{n\in\N}$ ensured by \eqref{eqn:approx:nu}). Since $\sum_{\alpha\in[T]}u^\alpha(\hat{t})=0$, we conclude that
\[ \P(\vert X_n\vert>\epsilon)\leq\P\Bigg(Y_n>\frac{1}{2}\min_{\tau\in[0,\hat{\tau}-\epsilon]}\sum_{\alpha\in[T]}\nu^\alpha(\tau)\Bigg)\to0,\quad\text{as }n\to\infty. \]

In order to prove the second part, we first want to show that the existence of $\bm{v}$ implies that in fact $\bm{z}(\hat{\tau})=\hat{\bm{z}}$. To this end, assume that $\bm{z}(\hat{\tau})\neq\hat{\bm{z}}$. Then there exists $(r,\alpha,\beta)\in V$ and $\delta>0$ such that $z^{r,\alpha,\beta}(\hat{\tau})>\hat{z}^{r,\alpha,\beta}+\delta v^{r,\alpha,\beta}$. Without loss of generality assume that $z^{r,\alpha,\beta}(\tau)$ is the first coordinate that reaches $\hat{z}^{r,\alpha,\beta}+\delta v^{r,\alpha,\beta}$, that is there exists $\tau_\delta\in[0,\hat{\tau}]$ such that $\bm{z}(\tau_\delta)\leq \hat{\bm{z}}+\delta\bm{v}$ componentwise and $z^{r,\alpha,\beta}(\tau_\delta)=\hat{z}^{r,\alpha,\beta}+\delta v^{r,\alpha,\beta}$. But by $D_{\bm{v}}f^{r,\alpha,\beta}(\hat{\bm{z}})<0$ and continuity of $D_{\bm{v}}f^{r,\alpha,\beta}(\bm{z})$, we then derive for $\delta>0$ small enough that $0 > f^{r,\alpha,\beta}(\hat{\bm{z}}+\delta\bm{v}) \geq f^{r,\alpha,\beta}(\bm{z}(\tau_\delta))$, where we used monotonicity of $f^{r,\alpha,\beta}$ from Lemma \ref{lem:properties:f}. This contradicts that $f^{r,\alpha,\beta}(\bm{z}(\tau))\geq0$ for all $\tau\in[0,\hat{\tau}]$ and hence it must hold that $\bm{z}(\hat{\tau})=\hat{\bm{z}}$. In particular, 
$g(\hat{\bm{z}})=\hat{\tau}$ (cf.~\eqref{eqn:tau:hat}).

The difficulty in the following is that the system is only described by the functions $\gamma_{j,m,l}^\alpha(\tau)$, $\nu^\alpha(\tau)$ and $\mu^{r,\alpha,\beta}(\tau)$ as long as $\tau<\hat{\tau}$, the first time at which $\sum_{\alpha\in[T]}\nu^\alpha(\tau)=0$. Wormald's theorem makes no statement about the system at or after $\hat{\tau}$, however. The idea is hence the following: We let $\tau_\epsilon$ be the first time at which $\mu^{r,\alpha,\beta}(\tau)\leq v^{r,\alpha,\beta}\epsilon$ for all $(r,\alpha,\beta)\in V$ and choose a sequence $(\epsilon_n)_{n\in\N}\subset\R_+$ such that $\epsilon_n\to0$ as $n\to\infty$. We then consider the cascade process as before for the first $\lfloor\tau_{\epsilon_n}n\rfloor$ steps and we show that the number of remaining defaults $R_n$ divided by $n$ converges to $0$ in probability as $n\to\infty$. In particular, this will show that

\vspace*{-0.6cm}
\[ n^{-1}\vert\mathcal{D}_n\vert = n^{-1}\hat{t} = n^{-1}(\lfloor\tau_{\epsilon_n}n\rfloor+R_n) \leq \tau_{\epsilon_n} + n^{-1}R_n \leq \hat{\tau} + o_p(1) = g(\hat{\bm{z}}) + o_p(1). \]

\vspace*{-0.1cm}
In order to show $n^{-1}R_n=o_p(1)$, we will expose the defaulted banks round by round as in \eqref{eqn:default:contagion}, i.\,e.~we expose the banks in $\bigcup_{\alpha\in[T]}U^\alpha(\lfloor\tau_{\epsilon_n}n\rfloor)$ at once and so on. However, banks with 
$w_i^{+,r,\alpha}=0$ for all $r\in[R]$ and $\alpha\in[T]$ will never infect any new banks. Thus, we only need to consider banks with $\sum_{r\in[R]}\sum_{\alpha\in[T]}w_i^{+,r,\alpha}>0$ in the following. Since we are in a finitary setting, this means that there exists $w_0>0$ such that $\sum_{r\in[R]}\sum_{\alpha\in[T]}w_i^{+,r,\alpha}\geq w_0$ for all banks. Taking into account also banks with total out-weight of \mbox{zero only causes an extra bounded factor for $R_n$}.

For each solvent bank at step $\lfloor\tau_{\epsilon_n}n\rfloor$ there are two possible ways to default: Either there is one exposure to a defaulted bank that is larger than the remaining capital at step $\lfloor\tau_{\epsilon_n}n\rfloor$ (the bank defaults \emph{directly}) or there are at least two exposures to defaulted banks that add up to an amount larger than the remaining capital (the bank defaults \emph{indirectly}). Therefore, for $\alpha\in[T]$ and $l\geq1$ we define the following sets:

\vspace*{-0.5cm}
\begin{align*}
\mathcal{D}_l^\alpha &= \mathcal{D}_l^\alpha(\tau_{\epsilon_n}) = \{i\in[n]\,:\,\alpha_i=\alpha\text{ and }i\text{ defaults directly in the }l\text{-th round after step }\lfloor\tau_{\epsilon_n}n\rfloor\}\\
\mathcal{I}_l^\alpha &= \mathcal{I}_l^\alpha(\tau_{\epsilon_n}) = \{i\in[n]\,:\,\alpha_i=\alpha\text{ and }i\text{ defaults indirectly in the }l\text{-th round after step }\lfloor\tau_{\epsilon_n}n\rfloor\}
\end{align*}
Further, let $\mathcal{T}^\alpha_l = \mathcal{D}^\alpha\cup\mathcal{I}^\alpha$. In particular, $R_n=\sum_{\alpha\in[T]}\sum_{l\geq1}\left\vert\mathcal{T}_l^\alpha(\tau_{\epsilon_n})\right\vert$. Further, the following quantities will play an important role:
\[ D_l^{r,\alpha,\beta} = \sum_{i\in\mathcal{D}_l^\beta}w_i^{+,r,\alpha}, I_l^{r,\alpha,\beta} = \sum_{i\in\mathcal{I}_l^\beta}w_i^{+,r,\alpha}\quad\text{and}\quad T_l^{r,\alpha,\beta}=\sum_{i\in\mathcal{T}_l^\beta}w_i^{+,r,\alpha},\quad l\geq 1, (r,\alpha,\beta)\in V \]

We now exploit again the assumption that $D_{\bm{v}}f^{r,\alpha,\beta}(\hat{\bm{z}})<0$ for all $(r,\alpha,\beta)\in V$. Also recall that the expression $D_{\bm{v}}f^{r,\alpha,\beta}(\bm{z})$ is continuous in $\bm{z}$ since the weights are assumed finitary. Further $\bm{z}(\tau)$ is continuous in $\tau$. Hence for $\epsilon>0$ small enough (i.\,e.~$\bm{z}(\tau_\epsilon)$ close to $\hat{\bm{z}}$), it holds
\begin{align*}
0 &> D_{\bm{v}}f^{r,\alpha,\beta}(\bm{z}(\tau_\epsilon))\\
&= \sum_{r'\in[R]} \E\Bigg[W^{+,r,\alpha}\Bigg(\sum_{\beta'\in[T]}v^{r',\beta,\beta'}W^{-,r',\beta'}\Bigg)\\
&\hspace{1.76cm}\times\P\Bigg(\sum_{s\in[R]}s\mathrm{Poi}\Bigg(\sum_{\gamma\in[T]}W^{-,s,\gamma}z^{s,\beta,\gamma}(\tau_\epsilon)\Bigg)\in\{C-r',\ldots,C-1\}\Bigg)\1\{A=\beta\}\Bigg] - v^{r,\alpha,\beta}\\
&= \sum_{j,m}\tilde{w}_j^{+,r,\alpha}\sum_{r'\in[R]}\Bigg(\sum_{\beta'\in[T]}v^{r',\beta,\beta'}\tilde{w}_j^{-,r',\beta'}\Bigg)\sum_{s=1}^{r'}\gamma_{j,m,m-s}^\beta(\tau_\epsilon) - v^{r,\alpha,\beta}.
\end{align*}
We can hence find $c_1<1$ such that for all $(r,\alpha,\beta)\in V$ it holds
\[ \sum_{j,m}\tilde{w}_j^{+,r,\alpha}\sum_{r'\in[R]}\Bigg(\sum_{\beta'\in[T]}v^{r',\beta,\beta'}\tilde{w}_j^{-,r',\beta'}\Bigg)\sum_{s=1}^{r'}\gamma_{j,m,m-s}^\beta(\tau_\epsilon) \leq c_1v^{r,\alpha,\beta}. \]
By \eqref{eqn:approx:gamma} and possibly slightly increasing $c_1$, we then derive that
\[ \sum_{j,m}\tilde{w}_j^{+,r,\alpha}\sum_{r'\in[R]}\Bigg(\sum_{\beta'\in[T]}v^{r',\beta,\beta'}\tilde{w}_j^{-,r',\beta'}\Bigg)\sum_{s=1}^{r'}\frac{c_{j,m,m-s}^\beta(\lfloor \tau_\epsilon n\rfloor)}{n} \leq c_1v^{r,\alpha,\beta} \]
on a $\sigma(h(\lfloor\tau_\epsilon n\rfloor))$-measurable set $\Omega_n^\epsilon$ such that $\lim_{n\to\infty}\P(\Omega_n^\epsilon)=1$ for every $\epsilon>0$. Further, by the definition of $\tau_\epsilon$ and \eqref{eqn:approx:mu}, we can choose $\Omega_n^\epsilon$ in such a way that $n^{-1}w^{r,\alpha,\beta}(\lfloor\tau_\epsilon n\rfloor)\leq2\epsilon v^{r,\alpha,\beta}$ holds on $\Omega_n^\epsilon$ for all $(r,\alpha,\beta)\in V$. We can then compute on $\Omega_n^\epsilon$

\vspace*{-0.255cm}
\begin{align*}
n^{-1}\E\Big[D_1^{r,\alpha,\beta}\,\Big\vert\,h(\lfloor\tau_\epsilon n\rfloor)\Big] 
&\leq \sum_{j,m} \tilde{w}_j^{+,r,\alpha}  \sum_{l=0}^{m-1}\frac{c_{j,m,l}^\beta(\lfloor\tau_\epsilon n\rfloor)}{n} \sum_{\beta'\in[T]}\sum_{v\in U^{\beta'}(\lfloor\tau_\epsilon n\rfloor)}\sum_{r'=m-l}^R\frac{w_v^{r',+,\beta}\tilde{w}_j^{-,r',\beta'}}{n}\\
&\leq 2\epsilon\sum_{j,m} \tilde{w}_j^{+,r,\alpha} \sum_{l=0}^{m-1}\frac{c_{j,m,l}^\beta(\lfloor\tau_\epsilon n\rfloor)}{n} \sum_{\beta'\in[T]}\sum_{r'=m-l}^R v^{r',\beta,\beta'}\tilde{w}_j^{-,r',\beta'}
\leq 2\epsilon c_1v^{r,\alpha,\beta},\\
n^{-1}\E\Big[I_1^{r,\alpha,\beta}\,\Big\vert\,h(\lfloor\tau_\epsilon n\rfloor)\Big] 
&\leq n^{-1}\sum_{j,m}\sum_{l=0}^{m-1}\sum_{i\in S_{j,m,l}^\beta(\lfloor\tau_\epsilon n\rfloor)} \overline{w} \Bigg( \sum_{\beta'\in[T]}\sum_{v\in U^{\beta'}(\lfloor\tau_\epsilon n\rfloor)}\sum_{r'=m-l}^R\frac{w_v^{r',+,\beta}w_i^{-,r',\beta'}}{n} \Bigg)^2\\
&\leq \overline{w}\Bigg( \sum_{\beta'\in[T]}\sum_{r'=0}^R\frac{w^{r',\beta,\beta'}(\lfloor\tau_\epsilon n\rfloor)}{n}\overline{w} \Bigg)^2
\leq 
C \epsilon^2,
\end{align*}
where $C:=4T^2(R+1)^2\overline{w}^3\Vert\bm{v}\Vert_\infty^2$. In particular, for $\epsilon>0$ small enough we find $c_2\in(0,1-c_1)$ such that $C\epsilon^2\leq 2\epsilon c_2 v^{r,\alpha,\beta}$ 
for all $(r,\alpha,\beta)\in V$ and hence on $\Omega_n^\epsilon$ it holds that
\[ n^{-1}\E\Big[T_1^{r,\alpha,\beta}\,\Big\vert\,h(\lfloor\tau_\epsilon n\rfloor)\Big] = n^{-1}\E\Big[D_1^{r,\alpha,\beta}\,\Big\vert\,h(\lfloor\tau_\epsilon n\rfloor)\Big] + n^{-1}\E\Big[I_1^{r,\alpha,\beta}\,\Big\vert\,h(\lfloor\tau_\epsilon n\rfloor)\Big] \leq 2\epsilon (c_1+c_2)v^{r,\alpha,\beta}. \]
Let then $c:=c_1+c_2\in(0,1)$. We continue inductively: Assume that on $\Omega_n^\epsilon$ it holds for $l\geq1$ that $n^{-1}\E\left[\left.T_l^{r,\alpha,\beta}\,\right\vert\,h(\lfloor\tau_\epsilon n\rfloor)\right] \leq 2\epsilon c^l v^{r,\alpha,\beta}$. We then derive on $\Omega_n^\epsilon$ that
\begin{align*}
n^{-1}\E\left[\left.D_{l+1}^{r,\alpha,\beta}\,\right\vert\,h(\lfloor\tau_\epsilon n\rfloor)\right] 
&= n^{-1}\sum_{j,m}\sum_{l=0}^{m-1}\sum_{i\in S_{j,m,l}^\beta(\lfloor\tau_\epsilon n\rfloor)} w_i^{+,r,\alpha}\P\left(\left.i\in\mathcal{D}_{l+1}^\beta\,\right\vert\,h(\lfloor\tau_\epsilon n\rfloor)\right)\\
&\hspace*{-1.65cm}\leq \sum_{j,m} \tilde{w}_j^{+,r,\alpha}  \sum_{l=0}^{m-1}\frac{c_{j,m,l}^\beta(\lfloor\tau_\epsilon n\rfloor)}{n} \E\left[\left.\sum_{\beta'\in[T]}\sum_{v\in \mathcal{T}_l^{\beta'}}\sum_{r'=m-l}^R\frac{w_v^{r',+,\beta}\tilde{w}_j^{-,r',\beta'}}{n}\,\right\vert\,h(\lfloor\tau_\epsilon n\rfloor)\right]\\
&\hspace*{-1.65cm}= \sum_{j,m}\tilde{w}_j^{+,r,\alpha}\sum_{l=0}^{m-1}\frac{c_{j,m,l}^\beta(\lfloor\tau_\epsilon n\rfloor)}{n}\sum_{\beta'\in[T]}\sum_{r'=m-l}^R\tilde{w}_j^{-,r',\beta'}n^{-1}\E\left[\left.T_l^{r',\beta,\beta'}\,\right\vert\,h(\lfloor\tau_\epsilon n\rfloor)\right]\\
&\hspace*{-1.65cm}\leq 
2\epsilon c^l \sum_{j,m}\tilde{w}_j^{+,r,\alpha}\sum_{r'\in[R]}\Bigg(\sum_{\beta'\in[T]}v^{r',\beta,\beta'}\tilde{w}_j^{-,r',\beta'}\Bigg)\sum_{s=1}^{r'}\frac{c_{j,m,m-s}^\beta(\lfloor\tau_\epsilon n\rfloor)}{n}
\leq 2\epsilon c^l c_1v^{r,\alpha,\beta},
\end{align*}
\begin{align*}
n^{-1}\E\left[\left.I_{l+1}^{r,\alpha,\beta}\,\right\vert\,h(\lfloor\tau_\epsilon n\rfloor)\right] 
&= n^{-1}\sum_{j,m}\sum_{l=0}^{m-1}\sum_{i\in S_{j,m,l}^\beta(\lfloor\tau_\epsilon n\rfloor)} w_i^{+,r,\alpha}\P\left(\left.i\in\mathcal{I}_{l+1}^\beta\,\right\vert\,h(\lfloor\tau_\epsilon n\rfloor)\right)\\
&\hspace{-1.3cm}\leq 
\sum_{j,m} \tilde{w}_j^{+,r,\alpha}  \sum_{l=0}^{m-1}\frac{c_{j,m,l}^\beta(\lfloor\tau_\epsilon n\rfloor)}{n} \left(\sum_{\beta'\in[T]}\sum_{r'=m-l}^R \tilde{w}_j^{-,r',\beta'}n^{-1}\E\left[\left.T_l^{r',\beta,\beta'}\,\right\vert\,h(\lfloor\tau_\epsilon n\rfloor)\right]\right)\\
&\hspace{2.0cm}\times \left(\sum_{\beta'\in[T]}\sum_{r'=m-l}^R\tilde{w}_j^{-,r',\beta'}n^{-1}\sum_{k\leq l}\E\left[\left.T_k^{r',\beta,\beta'}\,\right\vert\,h(\lfloor\tau_\epsilon n\rfloor)\right]\right)\\
&\hspace{-1.3cm}\leq \overline{w}\left(2\epsilon T(R+1)\overline{w} c^l v^{r',\beta,\beta'}\right)\Bigg(2\epsilon T(R+1)\overline{w}\sum_{k\leq l} c^k v^{r',\beta,\beta'}\Bigg) \leq C c^l \frac{1}{1-c} \epsilon^2
\end{align*}
Now choose $\epsilon>0$ small enough such that even $\frac{C}{1-c}\epsilon^2 \leq 2\epsilon c_2v^{r,\alpha,\beta}$ and conclude that on $\Omega_n^\epsilon$
\begin{align*}
&n^{-1}\E\left[\left.T_{l+1}^{r,\alpha,\beta}\,\right\vert\,h(\lfloor\tau_\epsilon n\rfloor)\right] = n^{-1}\E\left[\left.D_{l+1}^{r,\alpha,\beta}\,\right\vert\,h(\lfloor\tau_\epsilon n\rfloor)\right] + n^{-1}\E\left[\left.I_{l+1}^{r,\alpha,\beta}\,\right\vert\,h(\lfloor\tau_\epsilon n\rfloor)\right] \leq 2\epsilon c^{l+1} v^{r,\alpha,\beta},
\end{align*}
\begin{align*}
n^{-1}\sum_{\alpha\in[T]}\sum_{l\geq1}\E\left[\left.\vert\mathcal{T}_l^\alpha\vert\,\right\vert\,h(\lfloor\tau_\epsilon n\rfloor)\right] &\leq n^{-1}\sum_{\alpha\in[T]}\sum_{l\geq1}\E\Bigg[\sum_{i\in\mathcal{T}_l^\alpha} \frac{\sum_{r\in[R]}\sum_{\gamma\in[T]}w_i^{+,r,\gamma}}{w_0}\,\Bigg\vert\,h(\lfloor\tau_\epsilon n\rfloor)\Bigg]\\
&=w_0^{-1}n^{-1}\sum_{r\in[R]}\sum_{\alpha,\gamma\in[T]}\sum_{l\geq1}\E\left[\left.T_l^{r,\gamma,\alpha}\,\right\vert\,h(\lfloor\tau_\epsilon n\rfloor)\right]\\
&\leq 2\epsilon w_0^{-1}(R+1)T^2\frac{1}{1-c}\Vert\bm{v}\Vert_\infty.
\end{align*}
Consider now again the sequence $(\epsilon_n)_{n\in\N}$ from before and let $\epsilon,\delta>0$ arbitrary. For $n$ large enough such that $\epsilon_n\leq\epsilon$, we derive (using Markov's inequality in the penultimate step)
\begin{align*}
\P\left(n^{-1}R_n\geq\delta\right) &
\leq \P\Bigg(n^{-1} \sum_{\alpha\in[T]}\sum_{l\geq1}\vert\mathcal{T}_l^\alpha(\tau_\epsilon)\vert\geq\delta\Bigg)
\leq \E\Bigg[\delta^{-1}n^{-1}\sum_{\alpha\in[T]}\sum_{l\geq1}\E\left[\vert\mathcal{T}_l^\alpha(\tau_\epsilon)\vert\,\big\vert\,h(\lfloor\tau_\epsilon n\rfloor)\right]\Bigg]\\
&\leq 2\epsilon \delta^{-1}w_0^{-1}(R+1)T^2\frac{1}{1-c}\Vert\bm{v}\Vert_\infty + (1-\P(\Omega_n^\epsilon)).
\end{align*}
Choosing $\epsilon$ small enough and $n$ large enough, this quantity becomes arbitrarily small. Hence, $n^{-1}R_n=o_p(1)$ and this finishes the proof as explained above. 
\end{proof}

\begin{theorem}\label{thm:finitary:weights}
Consider a financial system described by a finitary regular vertex sequence and let $\hat{\bm{z}}$ and $\bm{z}^*$ be the smallest resp.~largest joint root in $S_0$ of the functions $\{f^{r,\alpha,\beta}\}_{(r,\alpha,\beta)\in V}$. Then $g(\hat{\bm{z}}) + o_p(1) \leq n^{-1}\vert\mathcal{D}_n\vert \leq g(\bm{z}^*) + o_p(1)$. In particular, if $\hat{\bm{z}}=\bm{z}^*$, then $n^{-1}\vert\mathcal{D}_n\vert = g(\hat{\bm{z}}) + o_p(1)$.
\end{theorem}

The idea for the proof of Theorem \ref{thm:finitary:weights} is to apply a small further shock to the financial system such that the second statement in Proposition \ref{prop:finitary:weights} becomes applicable. That is, we let each solvent bank in the system default independently with probability $\epsilon>0$ and denote the analogues of $f^{r,\alpha,\beta}$, $g$, $\hat{\bm{z}}$ and $\bm{z}^*$ by $f_\epsilon^{r,\alpha,\beta}$, $g_\epsilon$, $\hat{\bm{z}}(\epsilon)$ respectively $\bm{z}^*(\epsilon)$. That is,
\begin{align*}
f_\epsilon^{r,\alpha,\beta}(\bm{z}) &= \epsilon\left(\E\left[W^{+,r,\alpha}\1\{A=\beta\}\right] - z^{r,\alpha,\beta}\right) + (1-\epsilon)f^{r,\alpha,\beta}(\bm{z}),\hspace{1cm} & g_\epsilon(\bm{z}) &= \epsilon + (1-\epsilon)g(\bm{z}).
\end{align*}
We can assume in the following that $\E\left[W^{+,r,\alpha}\1\{A=\beta\}\right]>0$ and hence $f_\epsilon^{r,\alpha,\beta}(\bm{z})>f^{r,\alpha,\beta}(\bm{z})$ for all $\bm{z}$, since otherwise $f^{r,\alpha,\beta}(\bm{z})=-z^{r,\alpha,\beta}$ and we can simply leave out the $(r,\alpha,\beta)$-component in the proof. The following lemma describes $\bm{z}^*(\epsilon)$ for small $\epsilon$:
\begin{lemma}\label{lem:z^*:epsilon}
The function $\bm{z}^*:\R_{+,0}\to\R_{+,0}^V$ is right-continuous and monotonically increasing in each component. In particular, the derivative $(\bm{z}^*)'(\epsilon)$ exists for Lebesgue-almost every $\epsilon>0$ and $\bm{z}^*(\epsilon)-\bm{z}^*\geq\int_0^\epsilon(\bm{z}^*)'(\xi)\dd\xi$ componentwise.
\end{lemma}
\begin{proof}
For every $\bm{z}\in S_0$ and $\epsilon>0$ it holds $f_\epsilon^{r,\alpha,\beta}(\bm{z})\geq f^{r,\alpha,\beta}(\bm{z})\geq0$ and hence $S_0\subseteq S_0(\epsilon)$, where $S_0(\epsilon)$ denotes the analogue of $S_0$ for the additionally shocked case. In particular, $\bm{z}^*\in S_0(\epsilon)$ and hence $\bm{z}^*\leq\bm{z}^*(\epsilon)$ componentwise. The same argument shows that $\bm{z}^*(\epsilon_1)\leq\bm{z}^*(\epsilon_2)$ for any $\epsilon_1\leq\epsilon_2$ and hence $\bm{z}^*(\epsilon)$ is monotonically increasing in each component.

In particular, $\lim_{\epsilon\to0+}\bm{z}^*(\epsilon)$ exists and $\lim_{\epsilon\to0+}\bm{z}^*(\epsilon)\in\bigcap_{\epsilon>0}S_0(\epsilon)$. Now let $\delta>0$. By continuity of $f_\epsilon^{r,\alpha,\beta}(\bm{z})$ with respect to $\epsilon$ and $\bm{z}$, it holds $f^{r,\alpha,\beta}(\bm{z})\geq f_\epsilon^{r,\alpha,\beta}(\bm{z})-\delta$ for $\epsilon$ small enough and all $\bm{z}$ in the compact set $[\bm{0},\bm{\zeta}]$
. Hence for $\bm{z}\in\bigcap_{\epsilon>0}S_0(\epsilon)$, we obtain $f^{r,\alpha,\beta}(\bm{z})\geq-\delta$ for every $\delta>0$ and so $\bigcap_{\epsilon>0}S_0(\epsilon)\subseteq S$. However, since $\bigcap_{\epsilon>0}S_0(\epsilon)$ is the intersection of a chain of connected, compact sets in the Hausdorff space $\R^V$, it is itself a connected, compact set. Since further $\bm{0}\in\bigcap_{\epsilon>0}S_0(\epsilon)$, we thus derive that $\bigcap_{\epsilon>0}S_0(\epsilon)=S_0$. That is, $\lim_{\epsilon\to0+}\bm{z}^*(\epsilon)\in S_0$ and hence $\lim_{\epsilon\to0+}\bm{z}^*(\epsilon)=\bm{z}^*$. 
The same arguments show that $\lim_{h\to0+}\bm{z}^*(\epsilon+h)=\bm{z}^*(\epsilon)$ for every $\epsilon>0$, hence proving right-continuity of $\bm{z}^*(\epsilon)$.

A classical result for derivatives of monotone functions (see \cite[Theorem 7.21]{Wheeden1977} for instance) then yields the existence of $(\bm{z}^*)'$ almost everywhere and
\[ \bm{z}^*(\epsilon)-\bm{z}^*\geq \lim_{h\to0-}\bm{z}^*(\epsilon+h) - \lim_{h\to0+}\bm{z}^*(h) \geq \int_0^\epsilon(\bm{z}^*)'(\xi)\dd\xi. \qedhere\]
\end{proof}

\begin{proof}[Proof of Theorem \ref{thm:finitary:weights}]
As outlined above, in order to reduce this general setting to the special case from Proposition \ref{prop:finitary:weights}, we apply an additional small shock to the system. 
That is, if we can find a vector $\bm{v}(\epsilon)\in\R_+$ such that $D_{\bm{v}(\epsilon)}f_\epsilon^{r,\alpha,\beta}(\hat{\bm{z}}(\epsilon))<0$ for all $(r,\alpha,\beta)\in V$, then applying Proposition \ref{prop:finitary:weights}, we derive for the final default fraction $n^{-1}\vert\mathcal{D}_n^\epsilon\vert$ in the additionally shocked system
\[ n^{-1}\vert\mathcal{D}_n^\epsilon\vert \leq g_\epsilon(\hat{\bm{z}}(\epsilon)) + o_p(1) \leq g_\epsilon(\bm{z}^*(\epsilon)) + o_p(1) \leq \epsilon +  g(\bm{z}^*(\epsilon)) + o_p(1). \]
Hence for arbitrary $\delta>0$ and $\epsilon=\epsilon(\delta)>0$ small enough, we derive
\[ \P\left(n^{-1}\vert\mathcal{D}_n\vert - g(\bm{z}^*) > \delta\right)  \leq \P\left(n^{-1}\vert\mathcal{D}_n^\epsilon\vert - (\epsilon+g(\bm{z}^*(\epsilon)) > \delta/2\right) \to 0,\quad\text{as }n\to\infty \]
and thus
\[ n^{-1}\vert\mathcal{D}_n\vert \leq g(\bm{z}^*) + o_p(1). \]

So let us show the existence of the vectors $\bm{v}(\epsilon)$: By Lemma \ref{lem:z^*:epsilon} we know that $(\bm{z}^*)'(\epsilon)$ exists almost everywhere and that $\int_0^\epsilon(\bm{z}^*)'(\xi)\dd\xi \leq \bm{z}^*(\epsilon)-\bm{z}^* < \infty$. By the integrability, we can hence find a sequence $(\epsilon_n)_{n\in\N}\subset(0,1)$ such that $\epsilon_n\to0$ as $n\to\infty$ and for each $\epsilon_n$ it holds
\[ (\bm{z}^*)'(\epsilon_n)< \epsilon_n^{-1} \left(\pmb\zeta-\bm{z}^*-\delta\1\right) \]
componentwise, where 
$0<\delta<\zeta^{r,\alpha,\beta}-(z^*)^{r,\alpha,\beta}$ for all $(r,\alpha,\beta)\in V$ and $\1=(1,\ldots,1)\in\R_+^V$. This bound can be achieved with the same $\epsilon_n$ for each component, considering the sum $\sum_{(r,\alpha,\beta)\in V}((z^*)^{r,\alpha,\beta})'(\xi)$ which is still integrable. With
\[ 0 = f_\epsilon^{r,\alpha,\beta}(\bm{z}^*(\epsilon)) = (1-\epsilon)f^{r,\alpha,\beta}(\bm{z}^*(\epsilon)) + \epsilon\left(\E\left[W^{+,r,\alpha}\1\{A=\beta\}\right] - (z^*)^{r,\alpha,\beta}(\epsilon)\right), \]
we then derive that
\begin{align*}
\frac{\dd}{\dd \epsilon}f^{r,\alpha,\beta}(\bm{z}^*(\epsilon))\big\vert_{\epsilon=\epsilon_n} 
&= -\frac{1}{(1-\epsilon_n)^2}\left(\zeta^{r,\alpha,\beta}-(z^*)^{r,\alpha,\beta}(\epsilon_n)\right) + \frac{\epsilon_n}{1-\epsilon_n}\frac{\dd}{\dd\epsilon}(z^*)^{r,\alpha,\beta}(\epsilon)\big\vert_{\epsilon=\epsilon_n}\\
&< -\frac{1}{1-\epsilon_n}\left(\zeta^{r,\alpha,\beta}-(z^*)^{r,\alpha,\beta}(\epsilon_n) - \epsilon_n\frac{\dd}{\dd\epsilon}(z^*)^{r,\alpha,\beta}(\epsilon)\big\vert_{\epsilon=\epsilon_n}\right)\\
&< -\frac{1}{1-\epsilon_n} \left((z^*)^{r,\alpha,\beta}+\delta-(z^*)^{r,\alpha,\beta}(\epsilon_n)\right)<0
\end{align*}
for $n$ large enough such that $(z^*)^{r,\alpha,\beta}(\epsilon_n)<(z^*)^{r,\alpha,\beta}+\delta$. On the other hand, however,
\[ \frac{\dd}{\dd \epsilon}f^{r,\alpha,\beta}(\bm{z}^*(\epsilon))\big\vert_{\epsilon=\epsilon_n} = D_{\bm{v}(\epsilon_n)}f^{r,\alpha,\beta}(\bm{z}^*(\epsilon_n)) \geq D_{\bm{v}(\epsilon_n)}f_{\epsilon_n}^{r,\alpha,\beta}(\bm{z}^*(\epsilon_n)), \]
where $\bm{v}(\epsilon_n):=(\bm{z}^*)'(\epsilon_n)$. Hence altogether,
\[ D_{\bm{v}(\epsilon_n)}f_{\epsilon_n}^{r,\alpha,\beta}(\bm{z}^*(\epsilon_n)) < 0. \]
In fact, it also holds that
\begin{align*}
v^{r,\alpha,\beta}(\epsilon_n)
&\geq \lim_{h\to0}\E\Bigg[W^{+,r,\alpha}\1\{A=\beta\}\P\Bigg(\sum_{s\in[R]}s\mathrm{Poi}\Bigg(\sum_{\gamma\in[T]}W^{-,s,\gamma}(z^*)^{s,\beta,\gamma}(\epsilon_n+h)\Bigg)< C\Bigg)\Bigg]\\
&\geq \E\Bigg[W^{+,r,\alpha}\1\{A=\beta\}\P\Bigg(\sum_{s\in[R]}s\mathrm{Poi}\Bigg(\sum_{\gamma\in[T]}W^{-,s,\gamma}\zeta^{s,\beta,\gamma}\Bigg)< C\Bigg)\Bigg] > 0,
\end{align*}
using that $f_{\epsilon_n}^{r,\alpha,\beta}(\bm{z}^*(\epsilon_n))=f_{\epsilon_n+h}^{r,\alpha,\beta}(\bm{z}^*(\epsilon_n+h))=0$.
The proof is hence finished, if $\hat{\bm{z}}(\epsilon_n)=\bm{z}^*(\epsilon_n)$. 
Otherwise, note the following: For each $\delta>0$ it holds $f_{\epsilon_n}^{r,\alpha,\beta}(\hat{\bm{z}}(\epsilon_n+\delta))<f_{\epsilon_n+\delta}^{r,\alpha,\beta}(\hat{\bm{z}}(\epsilon_n+\delta))=0$. By monotonicity of $f_{\epsilon_n}^{r,\alpha,\beta}$ from Lemma \ref{lem:properties:f}, we derive that $\bm{z}^*(\epsilon_n)\leq\hat{\bm{z}}(\epsilon_n+\delta)\leq\bm{z}^*(\epsilon_n+\delta)$. Hence as $\delta\to0$, using that $\bm{z}^*(\epsilon_n+\delta)\to\bm{z}^*(\epsilon_n)$ by Lemma \ref{lem:z^*:epsilon}, we conclude that also $\hat{\bm{z}}(\epsilon_n+\delta)\to\bm{z}^*(\epsilon_n)$ as $\delta\to0$. Thus we derive that for $\delta_n>0$ small enough, it holds $D_{\bm{v}(\epsilon_n)}f_{\epsilon_n+\delta_n}^{r,\alpha,\beta}(\hat{\bm{z}}(\epsilon_n+\delta_n))<0$ by continuity of $D_{\bm{v}}f_\epsilon^{r,\alpha,\beta}(\bm{z})$ w.\,r.\,t.~$\epsilon$ and $\bm{z}$. Hence apply Proposition \ref{prop:finitary:weights} to the financial systems 
additionally shocked by $\epsilon_n+\delta_n$ and choose vectors $\bm{v}(\epsilon_n)$ for the directional derivative.
\end{proof}

\subsection{Proof of Theorem \ref{thm:general:weights}}\label{ssec:proof:main:general}
In the previous section we derived an explicit asymptotic expression for the final default fraction if in our model we choose vertex-weights only from a finite set. While this gives a first important insight into the behavior of large financial networks, it is not possible to model 
heavy tailed degree distributions as observed for real financial networks by bounded vertex-weights. Theorem \ref{thm:general:weights} hence extends Theorem \ref{thm:finitary:weights} to the case of general (non-finitary) regular vertex sequences.

The outline for the rest of this section is the following: We want to approximate the general regular vertex sequence by two sequences of finitary vertex sequences such that one of them describes a system that experiences less defaults and the other one experiences more defaults. To this end, we first construct the corresponding limiting distribution functions $\{F_k^A\}_{k\in\N}$ respectively $\{F_k^B\}_{k\in\N}$ and then investigate the finitary systems with help of Theorem \ref{thm:finitary:weights}.

Let $D_\infty:=\left(\R_{+,0}^{[R]\times[T]}\right)^2\times\N_0\times[T]$ and for $(r,\alpha,\beta)\in V$, $(\bm{z},\bm{x},\bm{y},\ell,m)\in\R_{+,0}^V\times D_\infty$, let
\begin{align*}
h_1^{r,\alpha,\beta}(\bm{z},\bm{x},\bm{y},\ell,m) &:= y^{r,\alpha}\psi_\ell\Bigg(\sum_{\gamma\in[T]}x^{1,\gamma}z^{1,\beta,\gamma},\ldots,\sum_{\gamma\in[T]}x^{R,\gamma}z^{R,\beta,\gamma}\Bigg) \1\{m=\beta\},\\
h_2(\bm{z},\bm{x},\bm{y},\ell,m) &:= \sum_{\beta\in[T]}\psi_\ell\Bigg(\sum_{\gamma\in[T]}x^{1,\gamma}z^{1,\beta,\gamma},\ldots,\sum_{\gamma\in[T]}x^{R,\gamma}z^{R,\beta,\gamma}\Bigg)\1\{m=\beta\},
\end{align*}
where as before $\psi_\ell(x_1,\ldots,x_R) := \P(\sum_{r\in[R]}r\mathrm{Poi}(x_r)\geq \ell)$.
Note that although $D_\infty$ does not contain $\left(\R_{+,0}^{[R]\times[T]}\right)^2\times\{\infty\}\times[T]$, it holds that $f^{r,\alpha,\beta}(\bm{z})=\int_{D_\infty} h_1^{r,\alpha,\beta}(\bm{z},\bm{x},\bm{y},\ell,m) \dd F(\bm{x},\bm{y},\ell,m) - z^{r,\alpha,\beta}$ and $g(\bm{z})=\int_{D_\infty} h_2(\bm{z},\bm{x},\bm{y},\ell,m) \dd F(\bm{x},\bm{y},\ell,m)$
, for $F$ the limiting distribution of the weights, capital and type as given in Definition \ref{def:regular:vertex:sequence}.
This is because $\psi_\infty(x_1,\ldots,x_R)
=0$. Let then $Z:=[\bm{0},\pmb\zeta]$ 
and $H := \{h_2\}\cup\bigcup_{(r,\alpha,\beta)\in V}\{h_1^{r,\alpha,\beta}\}$.

As a first approximation of $F$, we choose the discretizations
\begin{align*}
F_j^A(\bm{x},\bm{y},\ell,m) &:= F\left(\frac{\lceil j\bm{x}\rceil}{j}, \frac{\lceil j\bm{y}\rceil}{j},\ell,m\right), & F_j^B(\bm{x},\bm{y},\ell,m) &:= F\left(\frac{\lfloor j\bm{x}\rfloor}{j}, \frac{\lfloor j\bm{y}\rfloor}{j},\ell,m\right)
\end{align*}
for $j\in\N$, where $\lceil\cdot\rceil$ and $\lfloor\cdot\rfloor$ shall be applied componentwise. That is, the sequences $\{F_j^A\}_{j\in\N}$ and $\{F_j^B\}_{j\in\N}$ approximate $F$ from above respectively below and the approximations become finer as $j$ increases. Since every $h\in H$ is continuous in $\bm{z}$, $\bm{x}$ and $\bm{y}$, it is easy to obtain (cf.~\cite{Detering2015a}) that for each $k\in\N$ there exists $j_k$ large enough such that for all $j\geq j_k$ it holds
\[ \left\vert\int_{D_k} h(\bm{z},\bm{x},\bm{y},\ell,m)\dd F_j^{A,B}(\bm{x},\bm{y},\ell,m) - \int_{D_k} h(\bm{z},\bm{x},\bm{y},\ell,m)\dd F(\bm{x},\bm{y},\ell,m)\right\vert \leq k^{-1} \]
simultaneously for all $\bm{z}\in Z$, where $D_k:=\{(\bm{x},\bm{y},\ell,m)\,:\,x^{r,\alpha}\leq k,y^{r,\alpha}\leq k, \ell\leq k\}\subset D_\infty$. We denote $\overline{F}_k^A:=F_{j_k}^A$ and $\overline{F}_k^B:=F_{j_k}^B$ in the following.

By construction, the distribution functions $\overline{F}_k^{A,B}$ clearly correspond to discrete weight sequences that can be obtained from the original regular vertex sequence by adjusting weights upward respectively downward. However, $\overline{F}_k^{A,B}$ (potentially) still assigns mass to infinitely many weights and capitals. For the case of $\overline{F}_k^A$, we can overcome this issue by setting
\[ F_k^A(\bm{x},\bm{y},\ell,m) := \begin{cases}\overline{F}_k^A(\bm{x}\wedge k,\bm{y}\wedge k,\ell\wedge k,m),&\text{if }\ell<\infty,\\1,&\text{else},\end{cases} \]
where $\cdot\wedge k$ denotes componentwise truncation at $k$. That is, if the capital or one of the weights of some bank in the system exceeds $k$ (we call this bank \emph{large} in the following), then in the approximating finitary system described by $F_k^A$, this bank's weights are all set to $0$ and its capital is increased to $\infty$ (cf.~\cite{Detering2015a} for a rigorous definition of the approximating vertex sequences). Note that the type of each bank stays the same. Clearly, this further reduces defaults in the system in the sense that if we couple the original system with the finitary approximating system, then the final default fraction $n^{-1}\left\vert\left(\mathcal{D}_k^A\right)_n\right\vert$ is stochastically dominated by $n^{-1}\left\vert\mathcal{D}_n\right\vert$ for all $k\in\N$.

If we wanted to apply exactly the same idea also to $\overline{F}_k^B$, we would need to set all weights of large banks to $\infty$, which is not possible by the definition of a finitary regular vertex sequence. Still it will be possible to adjust weights and capitals of large banks to finitely many values such that the final default fraction in the finitary approximating system stochastically dominates $n^{-1}\vert\mathcal{D}_n\vert$. To this end, let $\gamma_k^\beta := \int_{D_k^c}\1\{m=\beta\}\dd F(\bm{x},\bm{y},\ell,m)$, where $D_k^c:=D_\infty\backslash D_k$, and
\[ \overline{w}_k^{r,\alpha,\beta} := \begin{cases}2\left(\gamma_k^\beta\right)^{-1}\int_{D_k^c}y^{r,\alpha}\1\{m=\beta\}\dd F(\bm{x},\bm{y},\ell,m)\geq 2k,&\text{if }\gamma_k^\beta>0,\\2k,&\text{if }\gamma_k^\beta=0.\end{cases} \]
Let then $F_k^B$ be given by $\overline{F}_k^B$ on $D_k$. By this definition we know that it holds $F_k^B(k,\ldots,k,\beta)=F(k,\ldots,k,\beta)$ for each $\beta\in[T]$ and hence $\P(A=\beta, C<\infty)-\P(A_k^B=\beta,C_k^B<\infty)=\gamma_k^\beta$. We can thus assign mass $\gamma_k^\beta$ to the points $(\bm{0},\overline{\bm{w}}_k,0,\beta)$. That is, if a large bank of type $\beta$ originally has finite capital, then its approximated capital is set to $0$ (it initially defaults), its in-weights are set to $0$ and its out-weights are set to $\overline{w}_k^{r,\alpha,\beta}$ (again cf.~\cite{Detering2015a} for a rigorous definition of the approximating vertex sequences). As before, their type does not change. Finally, we assign the remaining mass $\P(A=\beta,C=\infty)$ to the points $(\bm{0},\bm{0},\infty,\beta)$ for each $\beta\in[T]$.

By construction, all large banks are initially defaulted in the approximating finitary system. Also all the weights of small banks are increased as compared to the original system. To show that there occur more defaults in the approximating system than in the original one (i.\,e.~$n^{-1}\left\vert\left(\mathcal{D}_k^B\right)_n\right\vert$ stochastically dominates $n^{-1}\vert\mathcal{D}_n\vert$), all that is left to show is that for each $r\in[R]$ the total $r$-out-weight of the large banks with respect to each type $\alpha\in[T]$ in the approximating system is larger than in the original one. But the total $r$-out-weight of the large banks with respect to type $\alpha$ is given by
\[ n\sum_{\beta\in[T]}\overline{w}_k^{r,\alpha,\beta}\left(\gamma_k^\beta+o_p(1)\right) = 2n \int_{D_k^c}y^{r,\alpha}\dd F(\bm{x},\bm{y},\ell,m) (1+o_p(1)) \]
in the approximating system, whereas for the original system it is
\[ n\int_{D_k^c}y^{r,\alpha}\dd F(\bm{x},\bm{y},\ell,m) (1+o_p(1)). \]
Hence for each small bank $i\in[n]$ the number of incoming $r$-edges from large banks in the original system is stochastically dominated by the corresponding number in the approximating system (for more details see \cite{Detering2015a}). In particular, the total exposure of $i$ to the set of large banks (the weighted sum of incoming edges) is stochastically dominated. This shows the following:
\begin{lemma}\label{lem:stochastic:domination}
Consider a regular vertex sequence and let sequences $\{F_k^A\}$ and $\{F_k^B\}$ be constructed as above. Further let $\left(\mathcal{D}_k^A\right)_n$ and $\left(\mathcal{D}_k^B\right)_n$ be the sets of finally defaulted banks in the finitary approximating systems. Then with $\preceq$ denoting stochastic domination it holds that
\[ n^{-1}\left\vert\left(\mathcal{D}_k^A\right)_n\right\vert \preceq n^{-1}\vert\mathcal{D}_n\vert \preceq n^{-1}\left\vert\left(\mathcal{D}_k^B\right)_n\right\vert. \]
\end{lemma}
We have hence bounded the final default fraction $n^{-1}\vert\mathcal{D}_n\vert$ from below and from above using finitary approximations. We now want to compute the precise final default fractions for these approximating systems using Theorem \ref{thm:finitary:weights}. Let
\begin{gather*}
\left(f_k^{A,B}\right)^{r,\alpha,\beta}(\bm{z}) = \int_{D_\infty} h_1^{r,\alpha,\beta}(\bm{z},\bm{x},\bm{y},\ell,m)\dd F_k^{A,B}(\bm{x},\bm{y},\ell,m) - z^{r,\alpha,\beta},\\
g_k^{A,B}(\bm{z}) = \int_{D_\infty} h_2(\bm{z},\bm{x},\bm{y},\ell,m)\dd F_k^{A,B}(\bm{x},\bm{y},\ell,m)
\end{gather*}
the corresponding analogues of $f^{r,\alpha,\beta}$ and $g$. Further, denote by $\hat{\bm{z}}_k^A$ and $(\bm{z}^*)_k^A$ resp.~$\hat{\bm{z}}_k^B$ and $(\bm{z}^*)_k^B$ the smallest and largest joint roots of all functions $\left(f_k^A\right)^{r,\alpha,\beta}$ resp.~$\left(f_k^B\right)^{r,\alpha,\beta}$, $(r,\alpha,\beta)\in V$. Then we derive the following result comparing these quantities to the original system:
\begin{lemma}\label{lem:convergence:g}
It holds $\liminf_{k\to\infty} g_k^A\left(\hat{\bm{z}}_k^A\right) \geq g(\hat{\bm{z}})$ and $\limsup_{k\to\infty} g_k^B\left(\left(\bm{z}^*\right)_k^B\right) \leq g(\bm{z}^*)$.
\end{lemma}
\begin{proof}
First note that uniformly for all $\bm{z}\in Z$ and $h\in H$ it holds
\begin{align*}
&\left\vert \int_{D_k} h(\bm{z},\bm{x},\bm{y},\ell,m) \dd F_k^{A,B}(\bm{x},\bm{y},\ell,m) - \int_{D_k} h(\bm{z},\bm{x},\bm{y},\ell,m) \dd F(\bm{x},\bm{y},\ell,m)\right\vert\\
&\hspace{3.2cm}= \left\vert \int_{D_k} h(\bm{z},\bm{x},\bm{y},\ell,m) \dd \overline{F}_k^{A,B}(\bm{x},\bm{y},\ell,m) - \int_{D_k} h(\bm{z},\bm{x},\bm{y},\ell,m) \dd F(\bm{x},\bm{y},\ell,m)\right\vert\\
&\hspace{3.2cm}\leq k^{-1} \to 0,\quad\text{as }k\to\infty.
\end{align*}
Since further $\int_{D_k^c}\dd F\to0$ and $\int_{D_k^c}y^{r,\beta}\1\{m=\alpha\}\dd F\to0$, as $k\to\infty$, and each $h\in H$ is bounded by the integrands $1$ or $y^{r,\beta}\1\{m=\alpha\}$, it holds uniformly for all $\bm{z}\in Z$ and $h\in H$ that $\int_{D_k^c}h\dd F\to0$.  Together with $\int_{D_k^c}h\dd F_k^A=0$, this implies
\begin{equation}\label{eqn:strong:G:convergence:A}
\int_{D_\infty}h(\bm{z},\bm{x},\bm{y},\ell,m)\dd F_k^A(\bm{x},\bm{y},\ell,m) - \int_{D_\infty} h(\bm{z},\bm{x},\bm{y},\ell,m)\dd F(\bm{x},\bm{y},\ell,m) = o(1)
\end{equation}
uniformly for all $\bm{z}\in Z$ and $h\in H$.

For $\{F_k^B\}_{k\in\N}$, we further need to consider the term $\int_{D_k^c}h(\bm{z},\bm{x},\bm{y},\ell,m)\dd F_k^B(\bm{x},\bm{y},\ell,m)$. Thus
\begin{align*}
\int_{D_k^c}\dd F_k^B(\bm{x},\bm{y},\ell,m) &= \sum_{\alpha\in[T]}\gamma_k^\alpha, & \int_{D_k^c}y^{r,\beta}\1\{m=\alpha\}\dd F_k^B(\bm{x},\bm{y},\ell,m) &= \overline{w}_k^{r,\alpha,\beta} \gamma_k^\alpha.
\end{align*}
All these quantities tend to $0$ as $k\to\infty$ (note that $\overline{w}_k^{r,\alpha,\beta}\gamma_k^\alpha=2\int_{D_k^c}y^{r,\beta}\1\{m=\alpha\}\dd F$ if $\gamma_k^\alpha>0$). Since each function $h\in H$ is bounded by one of the (finitely many) integrands from above, this implies that $\int_{D_k^c}h(\bm{z},\bm{x},\bm{y},\ell,m)\dd F_k^B(\bm{x},\bm{y},\ell,m) \to 0$, as $k\to\infty$, uniformly for all $\bm{z}\in Z$ and $h\in H$. Therefore we can conclude that also uniformly for all $\bm{z}\in Z$ and $h\in H$,
\begin{equation}\label{eqn:strong:G:convergence:B}
\int_{D_\infty}h(\bm{z},\bm{x},\bm{y},\ell,m)\dd F_k^B(\bm{x},\bm{y},\ell,m) - \int_{D_\infty} h(\bm{z},\bm{x},\bm{y},\ell,m)\dd F(\bm{x},\bm{y},\ell,m) = o(1).
\end{equation}
We now turn to the proof of the first statement: Let $\epsilon>0$ and define

\vspace*{-0.2cm}
\[ D_\epsilon:=\bigcap_{(r,\alpha,\beta)\in V}\{\bm{z}\in\R_{+,0}^V\,:\,f^{r,\alpha,\beta}(\bm{z})\in[0,\epsilon]\}. \]

\vspace*{-0.2cm}
\noindent Further let $\bm{z}_\epsilon\in\R_{+,0}^V$ be defined by $z_\epsilon^{r,\alpha,\beta}:=\inf_{\bm{z}\in D_\epsilon}z^{r,\alpha,\beta}$, $(r,\alpha,\beta)\in V$. Then in particular, $\bm{z}_\epsilon\leq\hat{\bm{z}}$ componentwise since $\hat{\bm{z}}\in D_\epsilon$. Further, $\bm{z}_\epsilon$ is clearly increasing componentwise as $\epsilon\to0$. Hence the limit $\tilde{\bm{z}}:=\lim_{\epsilon\to0}\bm{z}_\epsilon\leq\hat{\bm{z}}$ exists. 
Now note that for fixed $(r,\alpha,\beta)\in V$, by definition of $\bm{z}_\epsilon$, we find a sequence $(\bm{z}_n)_{n\in\N}\subset D_\epsilon$ such that $\lim_{n\to\infty}z_n^{r,\alpha,\beta}=z_\epsilon^{r,\alpha,\beta}$ and $\bm{z}_n\geq\bm{z}_\epsilon$ componentwise. By monotonicity and uniform continuity of $f^{r,\alpha,\beta}$ on $D_\epsilon$, we then get
\[ f^{r,\alpha,\beta}(\bm{z}_\epsilon) \leq f^{r,\alpha,\beta}(z_n^{1,1,1},\ldots,z_\epsilon^{r,\alpha,\beta},\ldots,z_n^{R,T,T}) = f^{r,\alpha,\beta}(\bm{z}_n)+o(n)\leq\epsilon + o(n) \]
and hence $f^{r,\alpha,\beta}(\bm{z}_\epsilon)\leq\epsilon$. Again by continuity of $f^{r,\alpha,\beta}$, we obtain $f^{r,\alpha,\beta}(\tilde{\bm{z}}) = \lim_{\epsilon\to0}f^{r,\alpha,\beta}(\bm{z}_\epsilon) \leq \lim_{\epsilon\to0}\epsilon=0$. Replacing $\hat{\bm{z}}$ by $\tilde{\bm{z}}$ in the proof of Lemma \ref{lem:existence:hatz}, we now get the existence of a joint root $\bar{\bm{z}}\leq\tilde{\bm{z}}$ of all the functions $f^{r,\alpha,\beta}$, $(r,\alpha,\beta)\in V$. Since $\hat{\bm{z}}$ is the smallest joint root by definition, it thus follows that $\tilde{\bm{z}}=\hat{\bm{z}}$. Now note that by \eqref{eqn:strong:G:convergence:A} for $k$ large enough we derive that $(f_k^A)^{r,\alpha,\beta}(\bm{z})\geq f^{r,\alpha,\beta}(\bm{z})-\epsilon$ for all $\bm{z}\in Z$. Further, by construction of $F_k^A$, it holds that $(f_k^A)^{r,\alpha,\beta}(\bm{z})\leq f^{r,\alpha,\beta}(\bm{z})$. In particular, we can conclude that $\hat{\bm{z}}_k^A\in D_\epsilon$ for $k$ large enough and hence $\hat{\bm{z}}_k^A\geq \bm{z}_\epsilon$. Thus, for each $\epsilon>0$ by \eqref{eqn:strong:G:convergence:A} we derive 
$\liminf_{k\to\infty}g_k^A(\hat{\bm{z}}_k^A) \geq \lim_{k\to\infty}g_k^A(\bm{z}_\epsilon) = g(\bm{z}_\epsilon)$. Finally, using continuity of $g$ and $\lim_{\epsilon\to0}\bm{z}_\epsilon=\hat{\bm{z}}$, we get the first statement:
\[ \liminf_{k\to\infty}g_k^A(\hat{\bm{z}}_k^A) \geq g(\hat{\bm{z}}) \]

If now as in the proof of Theorem \ref{thm:finitary:weights} $\bm{z}^*(\epsilon)$ is the largest joint root of the additionally shocked system, we derive by \eqref{eqn:strong:G:convergence:B} that for $k$ large enough it holds $\left(f_k^B\right)^{r,\alpha,\beta}(\bm{z}^*(\epsilon))\leq f^{r,\alpha,\beta}(\bm{z}^*(\epsilon))/2<0$ for all $(r,\alpha,\beta)\in V$ and hence $(\bm{z}^*)_k^B\leq\bm{z}^*(\epsilon)$ componentwise. (Assume $\E[W^{+,r,\alpha}\1\{A=\beta\}]>0$ for all $(r,\alpha,\beta)\in V$ such that $f^{r,\alpha,\beta}(\bm{z}^*(\epsilon))<0$. Otherwise, we can simply leave out the coordinates $z^{r,\alpha,\beta}$ in all the proof since the $(r,\alpha,\beta)$-coordinate of all joint roots will be $0$.) Again by \eqref{eqn:strong:G:convergence:B}, we then derive 
$\limsup_{k\to\infty}g_k^B\left((\bm{z}^*)_k^B\right) \leq \lim_{k\to\infty} g_k^B(\bm{z}^*(\epsilon)) = g(\bm{z}^*(\epsilon))$ and by letting $\epsilon\to0$,
\[ \limsup_{k\to\infty} g_k^B\left((\bm{z}^*)_k^B\right) \leq g(\bm{z}^*). \qedhere\]

\end{proof}
\begin{proof}[Proof of Theorem \ref{thm:general:weights}]
Let $\epsilon>0$. By Lemma \ref{lem:stochastic:domination}, we obtain
\[ \P\left(n^{-1}\vert\mathcal{D}_n\vert-g(\hat{\bm{z}})<-\epsilon\right)\leq \P\left(n^{-1}\left\vert\left(\mathcal{D}_k^A\right)_n\right\vert-g(\hat{\bm{z}})<-\epsilon\right). \]
Further, by Lemma \ref{lem:convergence:g}, for $k$ large enough, we have $g_k^A(\hat{z}_k^A)>g(\hat{\bm{z}})-\epsilon/2$ and hence
\[ \P\left(n^{-1}\vert\mathcal{D}_n\vert-g(\hat{\bm{z}})<-\epsilon\right) \leq \P\left(n^{-1}\left\vert\left(\mathcal{D}_k^A\right)_n\right\vert-g_k^A(\hat{\bm{z}}_k^A)<-\epsilon/2\right). \]
Applying Theorem \ref{thm:finitary:weights} to the finitary system, as $n\to\infty$ we derive $\P\left(n^{-1}\vert\mathcal{D}_n\vert-g(\hat{\bm{z}}\right)<-\epsilon) \to 0$, which shows the first part of the theorem.

Similarly, for the second part, by Lemma \ref{lem:stochastic:domination}
\[ \P\left(n^{-1}\vert\mathcal{D}_n\vert-g(\bm{z}^*)>\epsilon\right) \leq \P\left(n^{-1}\left\vert\left(\mathcal{D}_k^B\right)_n\right\vert-g(\bm{z}^*)>\epsilon\right) \]
and by Lemma \ref{lem:convergence:g}, for $k$ large enough it holds that $g_k^B((\bm{z}^*)_k^B)<g(\bm{z}^*)+\epsilon/2$. Hence an application of Theorem \ref{thm:finitary:weights} yields that

\vspace*{-0.5cm}
\begin{align*}
\P\left(n^{-1}\vert\mathcal{D}_n\vert-g(\bm{z}^*)>\epsilon\right) &\leq \P\left(n^{-1}\left\vert\left(\mathcal{D}_k^B\right)_n\right\vert-g_k^B((\bm{z}^*)_k^B)>\frac{\epsilon}{2}\right) \to 0,\quad\text{as }n\to\infty. \qedhere
\end{align*}
\end{proof}

\subsection{Proofs for Section \ref{sec:resilience}}\label{ssec:proofs:resilience}
\begin{proof}[Proof of Theorem \ref{thm:resilience}]
Let $\gamma\in(0,1)$ and define $(f^\gamma)^{r,\alpha,\beta}(\bm{z}):=(1-\gamma)f^{r,\alpha,\beta}(\bm{z})+\gamma\left(\zeta^{r,\alpha,\beta}-z\right)$. Further, let $S^\gamma:=\bigcap_{(r,\alpha,\beta)\in V}\{\bm{z}\in\R_{+,0}^V\,:\,(f^\gamma)^{r,\alpha,\beta}(\bm{z})\geq0\}$ and denote by $S_0^\gamma$ the largest connected component of $S^\gamma$ containing $\bm{0}$. Finally, define $\bm{z}^*(\gamma)$ by $(z^*)^{r,\alpha,\beta}(\gamma):=\sup_{\bm{z}\in S_0^\gamma}z^{r,\alpha,\beta}$. By the proof of Lemma \ref{lem:z^*:epsilon} we know that $\bm{z}^*(\gamma)\to\bm{0}$, as $\gamma\to0$, and hence also $g(\bm{z}^*(\gamma))\to0$, using continuity of $g$. Choose now $\gamma>0$ small enough such that $g(\bm{z}^*(\gamma))\leq\epsilon/3$ and $\delta>0$ small enough such that $(f^M)^{r,\alpha,\beta}(\bm{z})<(f^\gamma)^{r,\alpha,\beta}(\bm{z})$ uniformly for all $\bm{0}\leq\bm{z}\leq\bm{\zeta}$ componentwise and $\P(M=0)<\delta$. 
Then in particular $(\bm{z}^*)^M\leq\bm{z}^*(\gamma)$ and $g((\bm{z}^*)^M)\leq g(\bm{z}^*(\gamma)) \leq\epsilon/3$.

If we now possibly decrease $\delta$ such that $\delta\leq\epsilon/3$, then by Theorem \ref{thm:general:weights}, we derive for the final fraction of defaulted banks in the shocked system $n^{-1}\vert\mathcal{D}_n^M\vert$ that w.\,h.\,p.
\[ n^{-1}\vert\mathcal{D}_n^M\vert \leq g^M((\bm{z}^*)^M) + \epsilon/3 \leq g((\bm{z}^*)^M) + 2\epsilon/3 \leq \epsilon. \qedhere \]
\end{proof}

\begin{proof}[Proof of Lemma \ref{lem:z0}]
Let $S_0(\epsilon,I)$ denote the largest connected subset of
\[ S(\epsilon,I) := \bigcap_{(r,\alpha,\beta)\in V}\left\{\bm{z}\in\R_{+,0}^V\,:\,f^{r,\alpha,\beta}(\bm{z})\geq -\epsilon\1\{(r,\alpha,\beta)\in I\}\right\} \]
containing $\bm{0}$. Then by replacing $S_0$ in the proof of Lemma~\ref{lem:existence:hatz} with $S_0(\epsilon,I)$, we obtain existence of a smallest (componentwise) point $\hat{\bm{z}}(\epsilon,I)\in\R_{+,0}^V$ such that $f^{r,\alpha,\beta}(\hat{\bm{z}}(\epsilon,I))=-\epsilon$ for $(r,\alpha,\beta)\in I$ and $f^{r\alpha,\beta}(\hat{\bm{z}}(\epsilon,I))=0$ for $(r,\alpha,\beta)\in V\backslash I$. 
In particular, $\hat{\bm{z}}(\epsilon,I)\in S_0(\epsilon,I)$. Let now
\[ T(\epsilon,I) := \bigcap_{(r,\alpha,\beta)\in V}\left\{\bm{z}\in\R_{+,0}^V\,:\,f^{r,\alpha,\beta}(\bm{z})\leq -\epsilon\1\{(r,\alpha,\beta)\in I\}\right\}. \]
Then clearly $\hat{\bm{z}}(\epsilon,I)\in T(\epsilon,I)$. Further, in the proof of the existence of $\hat{\bm{z}}(\epsilon,I)$ we can 
use any upper bound $\bm{z}\in T(\epsilon,I)$, which shows that $\hat{\bm{z}}(\epsilon,I)\leq\bm{z}$ componentwise. In particular, 
$\hat{\bm{z}}(\epsilon,I)$ is monotone in $\epsilon$ and therefore $\tilde{\bm{z}}(I):=\lim_{\epsilon\to0+}\hat{\bm{z}}(\epsilon,I)$ exists.

Let now $\bar{\bm{z}}\in T(I)$ arbitrary. Then there exists a sequence $(\bm{z}_k)_{k\in\N}\subset\R_{+,0}^V$ with $f^{r,\alpha,\beta}(\bm{z}_k)<0$ for $(r,\alpha,\beta)\in I$ respectively $f^{r,\alpha,\beta}(\bm{z}_k)\leq0$ for $(r,\alpha,\beta)\in V\backslash I$ such that $\lim_{k\to\infty}\bm{z}_k=\bar{\bm{z}}$. By finiteness of $I$, we can then find $\epsilon_k>0$ such that $f^{r,\alpha,\beta}(\bm{z}_k)\leq-\epsilon_k\1\{(r,\alpha,\beta)\in I\}$ for any $(r,\alpha,\beta)\in V$ and $k\in\N$. In particular, $\bm{z}_k\in T(\epsilon_k,I)$ and hence $\bm{z}_k\geq \hat{\bm{z}}(\epsilon_k,I)\geq\tilde{\bm{z}}$. As $k\to\infty$, we can thus conclude that $\tilde{\bm{z}}\leq\bar{\bm{z}}$ for any $\bar{\bm{z}}\in T(I)$ and hence $\tilde{\bm{z}}\leq\bm{z}_0(I)$. On the other hand, $\tilde{\bm{z}}(I)\in T(I)$ by definition and therefore $\tilde{\bm{z}}(I)=\bm{z}_0(I)$.

Finally, note that $\bm{z}_0(I)=\lim_{\epsilon\to0+}\hat{\bm{z}}(\epsilon,I) \in \bigcap_{\epsilon>0}S_0(\epsilon,I)=S_0$, where the last equality follows from $\bigcap_{\epsilon>0}S_0(\epsilon,I)\subset S$ and that $\bigcap_{\epsilon>0}S_0(\epsilon,I)$ must be a connected set containing $\bm{0}$ since $S_0(\epsilon,I)$ is a chain of connected, compact sets containing $\bm{0}$.
\end{proof}

\begin{proof}[Proof of Theorem \ref{thm:non-resilience}]
Let $\hat{\bm{z}}^M$ denote the analogue of $\hat{\bm{z}}$ for the ex post shocked system. Then 
\begin{align*}
&f^{r,\alpha,\beta}(\hat{\bm{z}}^M) + \E\Bigg[W^{+,r,\alpha}\P\Bigg(\sum_{s\in[R]}s\mathrm{Poi}\Bigg(\sum_{\gamma\in[T]}W^{-,s,\gamma}(\hat{z}^M)^{s,\beta,\gamma}\Bigg)\leq C-1\Bigg)\1\{A=\beta\}\1\{M=0\}\Bigg]\\
&= (f^M)^{r,\alpha,\beta}(\hat{\bm{z}}^M) = 0
\end{align*}
and hence $f^{r,\alpha,\beta}(\hat{\bm{z}}^M)\leq0$ with equality if and only if $\E[W^{+,r,\alpha}\1\{A=\beta\}\1\{M=0\}]=0$. Define now $\epsilon:=-\max_{(r,\alpha,\beta)\in I}f^{r,\alpha,\beta}(\hat{\bm{z}}^M)>0$, so that $\hat{\bm{z}}^M\in T(\epsilon,I)$, where $T(\epsilon,I)$ as in the proof of Lemma \ref{lem:z0}. 
In the construction of $\hat{\bm{z}}(\epsilon,I)$ (see Lemma \ref{lem:z0}), we can then use the upper bound $\hat{\bm{z}}^M$ and obtain that 
$\hat{\bm{z}}(\epsilon,I)\leq\hat{\bm{z}}^M$ and hence $\hat{\bm{z}}^M\geq\bm{z}_0(I)$. Since now
\begin{align*}
\delta &:= g^M(\hat{\bm{z}}^M) - g(\bm{z}_0(I)) \geq g^M(\bm{z}_0(I)) - g(\bm{z}_0(I))\\
&\hphantom{:}=\sum_{\beta\in[T]}\E\Bigg[\P\Bigg(\sum_{s\in[R]}s\mathrm{Poi}\Bigg(\sum_{\gamma\in[T]}W^{-,s,\gamma}z_0^{s,\beta,\gamma}(I)\Bigg) \leq C-1\Bigg)\1\{A=\beta\}\1\{M=0\}\Bigg] >0,
\end{align*}
we can apply Theorem \ref{thm:general:weights} to conclude that
\[ \lim_{n\to\infty}\P\left(n^{-1}\vert\mathcal{D}_n^M\vert < g(\bm{z}_0(I))\right) = \lim_{n\to\infty} \P\left(n^{-1}\vert\mathcal{D}_n^M\vert < g^M(\hat{\bm{z}}^M) - \delta\right) = 0 \]
and hence $n^{-1}\vert\mathcal{D}_n^M\vert \geq g(\bm{z}_0(I))$ w.\,h.\,p.

Now assume that $\bm{z}_0(I)\neq\bm{0}$. Since $\bm{z}_0(I)$ is a joint root of the functions $f^{r,\alpha,\beta}$, $(r,\alpha,\beta)\in V$, we then derive that there exists $(r,\alpha,\beta)\in V$ such that
\[ \E\Bigg[W^{+,r,\alpha}\P\Bigg(\sum_{s\in[R]}s\mathrm{Poi}\Bigg(\sum_{\gamma\in[T]}W^{-,s,\gamma}z_0^{s,\beta,\gamma}(I)\Bigg)\geq C\Bigg) \1\{A=\beta\}\Bigg] = z_0^{r,\alpha,\beta}(I) > 0 \]
and hence
\[ \E\Bigg[\P\Bigg(\sum_{s\in[R]}s\mathrm{Poi}\Bigg(\sum_{\gamma\in[T]}W^{-,s,\gamma}z_0^{s,\beta,\gamma}(I)\Bigg)\geq C\Bigg) \1\{A=\beta\}\Bigg] > 0. \]
In particular, $g(\bm{z}_0(I))>0$.
\end{proof}

\begin{proof}[Proof of Lemma \ref{lem:z0:equals:z*}]
By Lemma \ref{lem:z0} clearly $\bm{z}_0(\tilde{V})\leq\bm{z}^*$. Assume now that $\bm{z}_0(\tilde{V})\lneq\bm{z}^*$. Then for some $(r,\alpha,\beta)\in\tilde{V}$ it must hold $z_0^{r,\alpha,\beta}(\tilde{V})<(z^*)^{r,\alpha,\beta}$ and by the construction of $\bm{z}_0(\tilde{V})$ in the proof of Lemma \ref{lem:z0} we can find $\epsilon>0$ such that $z_0^{r,\alpha,\beta}(\tilde{V})\leq\hat{z}^{r,\alpha,\beta}(\epsilon,\tilde{V})<(z^*)^{r,\alpha,\beta}$. Now by the definition of $\bm{z}^*$ and connectedness of $S_0$, we find $S_0\ni\tilde{\bm{z}}\leq\hat{\bm{z}}(\epsilon,\tilde{V})$ such that $\tilde{z}^{r,\alpha,\beta}=\hat{z}^{r,\alpha,\beta}(\epsilon,\tilde{V})$. But then $f^{r,\alpha,\beta}(\tilde{\bm{z}})\leq f^{r,\alpha,\beta}(\hat{\bm{z}}(\epsilon,\tilde{V}))=-\epsilon<0$ which contradicts $\tilde{\bm{z}}\in S_0$.
\end{proof}

{\begin{multicols}{2}
\footnotesize
\bibliography{finance}
\bibliographystyle{abbrv}
\end{multicols}}

\end{document}

%% file: definitions.tex
\usepackage{mathrsfs}

\newcommand{\E}{{\mathbb{E}}}
\providecommand{\R}{{\mathbb{R}}}

\newcommand{\dd}{{\rm d}}

\providecommand{\N}{{\mathbb N}}

\newcommand{\1}{\ensuremath{\mathbf{1}}}

\ifx\prop\undefined
\newtheorem{prop}{Proposition}[section]
\fi

\newtheorem{remark}[theorem]{Remark}

\def\P{{\mathbb P}}